\documentclass[11pt]{article}
\usepackage{amsmath}
\usepackage{epsfig}
\usepackage{amsthm}
\usepackage{graphicx}
\usepackage{color,multirow,array}

\usepackage{dcolumn}
\usepackage{bm}

\usepackage{amssymb}
\usepackage{bbm}

\newtheorem{lemma}{Lemma}[section]
\newtheorem{theorem}[lemma]{Theorem}
\newtheorem{corollary}[lemma]{Corollary}

\theoremstyle{definition}
\newtheorem{definition}[lemma]{Definition}

\theoremstyle{definition}
\newtheorem{observation}[lemma]{Observation}
\theoremstyle{definition}
\newtheorem{example}[lemma]{Example}

\theoremstyle{definition}
\newtheorem{construction}[lemma]{Construction}

\long\def\symbolfootnote[#1]#2{\begingroup%
\def\thefootnote{\fnsymbol{footnote}}\footnote[#1]{#2}\endgroup}

\newcommand{\beq}{\begin{equation}}
\newcommand{\eeq}{\end{equation}}
\newcommand{\bea}{\begin{eqnarray}}
\newcommand{\eea}{\end{eqnarray}}
\newcommand{\ena}{\end{eqnarray}}

\def\mathfrak{\bf}

\renewcommand{\[}{\left[}
\renewcommand{\]}{\right]}

\def\be{\begin{equation}}
\def\ee{\end{equation}}
\def\bea{\begin{eqnarray}}
\def\eea{\end{eqnarray}}
\def\dt#1{\on{\hbox{\bf .}}{#1}}                % (big) dot over
\def\Dot#1{\dt{#1}}
\def\IR{\relax{\rm I\kern-.18em R}}
\def\binomial#1#2{\left(\,{\buildrel
{\raise4pt\hbox{$\displaystyle{#1}$}}\over
{\raise-6pt\hbox{$\displaystyle{#2}$}}}\,\right)}

\def\[{\lfloor{\hskip 0.35pt}\!\!\!\lceil}
\def\]{\rfloor{\hskip 0.35pt}\!\!\!\rceil}

\catcode`@=11
\def\un#1{\relax\ifmmode\@@underline#1\else
        $\@@underline{\hbox{#1}}$\relax\fi}
\catcode`@=12

\def\ad{{\kern0.5pt
                   \alpha \kern-5.05pt
\raise5.8pt\hbox{$\textstyle.$}\kern
0.5pt}}

\def\Dot#1{{\kern0.5pt
     {#1} \kern-5.05pt \raise5.8pt\hbox{$\textstyle.$}\kern
0.5pt}}

\def\m{\mu}

\def\bo{{\raise.15ex\hbox{\large$\Box$}}}               % D'Alembertian
\def\TH{{\raise.2ex\hbox{$\displaystyle \bigodot$}\mskip-4.7mu \llap H
\;}}
\def\face{{\raise.2ex\hbox{$\displaystyle \bigodot$}\mskip-2.2mu \llap
{$\ddot
        \smile$}}}                                      % happy face
\def\leftrightarrowfill{$\mathsurround=0pt \mathord\leftarrow \mkern-6mu
        \cleaders\hbox{$\mkern-2mu \mathord- \mkern-2mu$}\hfill
        \mkern-6mu \mathord\rightarrow$}
\def\dvec#1{\vbox{\ialign{##\crcr
        \leftrightarrowfill\crcr\noalign{\kern-1pt\nointerlineskip}
        $\hfil\displaystyle{#1}\hfil$\crcr}}}           % <--> accent
\def\dt#1{{\buildrel {\hbox{\LARGE .}} \over {#1}}}     % dot-over forsp/sb
\def\sfrac#1#2{{\vphantom1\smash{\lower.5ex\hbox{\small$#1$}}\over
        \vphantom1\smash{\raise.4ex\hbox{\small$#2$}}}} % alternatefraction
\def\bfrac#1#2{{\vphantom1\smash{\lower.5ex\hbox{$#1$}}\over
        \vphantom1\smash{\raise.3ex\hbox{$#2$}}}}       % "
\def\afrac#1#2{{\vphantom1\smash{\lower.5ex\hbox{$#1$}}\over#2}}    % "
\def\on#1#2{\mathop{\null#2}\limits^{#1}}               % arbitraryaccent
\newskip\humongous \humongous=0pt plus 1000pt minus 1000pt

\newif\ifdtup

  \def\pp{{\mathchoice
          {%displaystyle
              \kern 1pt%
              \raise 1pt
              \vbox{\hrule width5pt height0.4pt depth0pt
                    \kern -2pt
                    \hbox{\kern 2.3pt
                          \vrule width0.4pt height6pt depth0pt
                          }
                    \kern -2pt
                    \hrule width5pt height0.4pt depth0pt}%
                    \kern 1pt
           }
            {%textstyle
              \kern 1pt%
              \raise 1pt
              \vbox{\hrule width4.3pt height0.4pt depth0pt
                    \kern -1.8pt
                    \hbox{\kern 1.95pt
                          \vrule width0.4pt height5.4pt depth0pt
                          }
                    \kern -1.8pt
                    \hrule width4.3pt height0.4pt depth0pt}%
                    \kern 1pt
            }
            {%scriptstyle
              \kern 0.5pt%
              \raise 1pt
              \vbox{\hrule width4.0pt height0.3pt depth0pt
                    \kern -1.9pt  %[e]=0.15pt
                    \hbox{\kern 1.85pt
                          \vrule width0.3pt height5.7pt depth0pt
                          }
                    \kern -1.9pt
                    \hrule width4.0pt height0.3pt depth0pt}%
                    \kern 0.5pt
            }
            {%scriptscriptstyle
              \kern 0.5pt%
              \raise 1pt
              \vbox{\hrule width3.6pt height0.3pt depth0pt
                    \kern -1.5pt
                    \hbox{\kern 1.65pt
                          \vrule width0.3pt height4.5pt depth0pt
                          }
                    \kern -1.5pt
                    \hrule width3.6pt height0.3pt depth0pt}%
                    \kern 0.5pt%}
            }
        }}

  \def\mm{{\mathchoice
                       {%displaystyle
                             \kern 1pt
               \raise 1pt    \vbox{\hrule width5pt height0.4pt depth0pt
                                  \kern 2pt
                                  \hrule width5pt height0.4pt depth0pt}
                             \kern 1pt}
                       {%textstyle
                            \kern 1pt
               \raise 1pt \vbox{\hrule width4.3pt height0.4pt depth0pt
                                  \kern 1.8pt
                                  \hrule width4.3pt height0.4pt depth0pt}
                             \kern 1pt}
                       {%scriptstyle
                            \kern 0.5pt
               \raise 1pt
                            \vbox{\hrule width4.0pt height0.3pt depth0pt
                                  \kern 1.9pt
                                  \hrule width4.0pt height0.3pt depth0pt}
                            \kern 1pt}
                       {%scriptscriptstyle
                           \kern 0.5pt
             \raise 1pt  \vbox{\hrule width3.6pt height0.3pt depth0pt
                                  \kern 1.5pt
                                  \hrule width3.6pt height0.3pt depth0pt}
                           \kern 0.5pt}
                       }}

\def\pd{{\kern0.5pt
                   + \kern-5.05pt \raise5.8pt\hbox{$\textstyle.$}\kern
0.5pt}}

\def\pmd{{\kern0.5pt
                  \pm \kern-5.05pt \raise6.3pt\hbox{$\textstyle.$}\kern1.5pt}}

\def\md{{\mathchoice
   {%displaystyle
      {{\kern 1pt - \kern-6.2pt \raise5pt\hbox{$\textstyle.$}\kern 1pt}}}
    {%textstyle
      {{\kern 1pt - \kern-6.2pt \raise5pt\hbox{$\textstyle.$}\kern 1pt}}}
    {%scriptstyle
      {\kern0.5pt - \kern-5.05pt \raise3.4pt\hbox{$\textstyle.$}\kern0.5pt}}
    {%scriptscriptstyle
      {\kern0.5pt - \kern-5.05pt \raise3.4pt\hbox{$\textstyle.$}\kern0.5pt}}}}

\def\ad{{\dot{\alpha}}}

\def\pp{{\mathchoice
          {%displaystyle
              \kern 1pt%
              \raise 1pt
              \vbox{\hrule width5pt height0.4pt depth0pt
                    \kern -2pt
                    \hbox{\kern 2.3pt
                          \vrule width0.4pt height6pt depth0pt
                          }
                    \kern -2pt
                    \hrule width5pt height0.4pt depth0pt}%
                    \kern 1pt
           }
            {%textstyle
              \kern 1pt%
              \raise 1pt
              \vbox{\hrule width4.3pt height0.4pt depth0pt
                    \kern -1.8pt
                    \hbox{\kern 1.95pt
                          \vrule width0.4pt height5.4pt depth0pt
                          }
                    \kern -1.8pt
                    \hrule width4.3pt height0.4pt depth0pt}%
                    \kern 1pt
            }
            {%scriptstyle
              \kern 0.5pt%
              \raise 1pt
              \vbox{\hrule width4.0pt height0.3pt depth0pt
                    \kern -1.9pt  %[e]=0.15pt
                    \hbox{\kern 1.85pt
                          \vrule width0.3pt height5.7pt depth0pt
                          }
                    \kern -1.9pt
                    \hrule width4.0pt height0.3pt depth0pt}%
                    \kern 0.5pt
            }
            {%scriptscriptstyle
              \kern 0.5pt%
              \raise 1pt
              \vbox{\hrule width3.6pt height0.3pt depth0pt
                    \kern -1.5pt
                    \hbox{\kern 1.65pt
                          \vrule width0.3pt height4.5pt depth0pt
                          }
                    \kern -1.5pt
                    \hrule width3.6pt height0.3pt depth0pt}%
                    \kern 0.5pt%}
            }
        }}

  \def\mm{{\mathchoice
                       {%displaystyle
                             \kern 1pt
               \raise 1pt    \vbox{\hrule width5pt height0.4pt depth0pt
                                  \kern 2pt
                                  \hrule width5pt height0.4pt depth0pt}
                             \kern 1pt}
                       {%textstyle
                            \kern 1pt
               \raise 1pt \vbox{\hrule width4.3pt height0.4pt depth0pt
                                  \kern 1.8pt
                                  \hrule width4.3pt height0.4pt depth0pt}
                             \kern 1pt}
                       {%scriptstyle
                            \kern 0.5pt
               \raise 1pt
                            \vbox{\hrule width4.0pt height0.3pt depth0pt
                                  \kern 1.9pt
                                  \hrule width4.0pt height0.3pt depth0pt}
                            \kern 1pt}
                       {%scriptscriptstyle
                           \kern 0.5pt
             \raise 1pt  \vbox{\hrule width3.6pt height0.3pt depth0pt
                                  \kern 1.5pt
                                  \hrule width3.6pt height0.3pt depth0pt}
                           \kern 0.5pt}
                       }}

\def\pd{{\kern0.5pt
                   + \kern-5.05pt \raise5.8pt\hbox{$\textstyle.$}\kern
0.5pt}}

\def\pmd{{\kern0.5pt
                  \pm \kern-5.05pt \raise6.3pt\hbox{$\textstyle.$}\kern1.5pt}}

\def\md{{\mathchoice
   {%displaystyle
      {{\kern 1pt - \kern-6.2pt \raise5pt\hbox{$\textstyle.$}\kern 1pt}}}
    {%textstyle
      {{\kern 1pt - \kern-6.2pt \raise5pt\hbox{$\textstyle.$}\kern 1pt}}}
    {%scriptstyle
      {\kern0.5pt - \kern-5.05pt \raise3.4pt\hbox{$\textstyle.$}\kern0.5pt}}
    {%scriptscriptstyle
      {\kern0.5pt - \kern-5.05pt \raise3.4pt\hbox{$\textstyle.$}\kern0.5pt}}}}

\def\dslash{\not{\hbox{\kern-2pt $\partial$}}}
\def\Dslash{\not{\hbox{\kern-4pt $D$}}}
\def\pslash{\not{\hbox{\kern-2.3pt $p$}}}
 \newtoks\slashfraction
 \slashfraction={.13}
 \def\slash#1{\setbox0\hbox{$ #1 $}
 \setbox0\hbox to \the\slashfraction\wd0{\hss \box0}/\box0 }

\font\ro=cmsy10                          % font with rope
\def\kcr{{\hbox{\ro \char'170}}}                % right-handed rope
\def\ktl{{\hbox{\ro \char'170}}}        % top end for left-handed rope
\def\ktr{{\hbox{\ro \char'170}}}        % " right
\def\kbl{{\hbox{\ro \char'170}}}        % " bottom left
\def\kbr{{\hbox{\ro \char'170}}}        % " right
\def\plpl{\raise-2pt\hbox{$\raise3pt\hbox{$_+$}\hskip-6.67pt\raise0.0pt
\hbox{$^+$}\hskip 0.01pt$}}
\def\mimi{\raise-2pt\hbox{$\raise3pt\hbox{$_-$}\hskip-6.67pt\raise0.0pt
\hbox{$^-$}\hskip 0.01pt$}}

\def\bo{{\raise.15ex\hbox{\large$\Box$}}}               % D'Alembertian
\def\TH{{\raise.2ex\hbox{$\displaystyle \bigodot$}\mskip-4.7mu \llap H \;}}
\def\face{{\raise.2ex\hbox{$\displaystyle \bigodot$}\mskip-2.2mu \llap {$\ddot
        \smile$}}}                                      % happy face
\def\leftrightarrowfill{$\mathsurround=0pt \mathord\leftarrow \mkern-6mu
        \cleaders\hbox{$\mkern-2mu \mathord- \mkern-2mu$}\hfill
        \mkern-6mu \mathord\rightarrow$}
\def\dvec#1{\vbox{\ialign{##\crcr
        \leftrightarrowfill\crcr\noalign{\kern-1pt\nointerlineskip}
        $\hfil\displaystyle{#1}\hfil$\crcr}}}           % <--> accent
\def\dt#1{{\buildrel {\hbox{\LARGE .}} \over {#1}}}     % dot-over for sp/sb
\def\sfrac#1#2{{\vphantom1\smash{\lower.5ex\hbox{\small$#1$}}\over
        \vphantom1\smash{\raise.4ex\hbox{\small$#2$}}}} % alternate fraction
\def\bfrac#1#2{{\vphantom1\smash{\lower.5ex\hbox{$#1$}}\over
        \vphantom1\smash{\raise.3ex\hbox{$#2$}}}}       % "
\def\afrac#1#2{{\vphantom1\smash{\lower.5ex\hbox{$#1$}}\over#2}}    % "
\def\on#1#2{\mathop{\null#2}\limits^{#1}}               % arbitrary accent
\parskip=\medskipamount                 % space between paragraphs (LaTeX)
\thicklines                         % thick straight lines for pictures (LaTeX)

\def\oldheadpic{                                % old UM heading
        \setlength{\unitlength}{.4mm}
        \thinlines
        \par
        \begin{picture}(349,16)
        \put(325,16){\line(1,0){4}}
        \put(330,16){\line(1,0){4}}
        \put(340,16){\line(1,0){4}}
        \put(335,0){\line(1,0){4}}
        \put(340,0){\line(1,0){4}}
        \put(345,0){\line(1,0){4}}
        \put(329,0){\line(0,1){16}}
        \put(330,0){\line(0,1){16}}
        \put(339,0){\line(0,1){16}}
        \put(340,0){\line(0,1){16}}
        \put(344,0){\line(0,1){16}}
        \put(345,0){\line(0,1){16}}
        \put(329,16){\oval(8,32)[bl]}
        \put(330,16){\oval(8,32)[br]}
        \put(339,0){\oval(8,32)[tl]}
        \put(345,0){\oval(8,32)[tr]}
        \end{picture}
        \par
        \thicklines
        \vskip.2in}
\def\oldtitle#1#2#3#4{\oldheadpic\begin{center}\vglue.5in{\large\bf #1}\\[.6in]
        {#2}\\[.1in] {\it Department of Physics and Astronomy}\\
        {\it University of Maryland, College Park, MD 20742}\\[.6in]
        Physics Publication \#{#3}\\ {#4}\\[1.5in] {\bf ABSTRACT}\\[.1in]
        \end{center} \begin{quotation}}                 % old title stuff
\def\oldTitle#1#2#3#4#5#6#7{\oldheadpic\begin{center} \vglue .4in
        {\large\bf #1}\\[.4in]
        {#2}\\[.1in] {\it Department of Physics and Astronomy}\\
        {\it University of Maryland, College Park, MD 20742}\\[.1in]
        {#3}\\[.1in] {\it {#4}}\\ {\it {#5}}\\[.4in]         
        Physics Publication \#{#6}\\ {#7}\\[.5in] {\bf ABSTRACT}\\[.1in]
        \end{center} \begin{quotation}}                 % " for 2 authors
\def\border{                                            % border
        \setlength{\unitlength}{1mm}
        \newcount\xco
        \newcount\yco
        \xco=-21
        \yco=12
        \begin{picture}(140,0)
        \put(\xco,\yco){$\ktl$}
        \advance\yco by-1
        {\loop
        \put(\xco,\yco){$\kcr$}
        \advance\yco by-2
        \ifnum\yco>-240
        \repeat
        \put(\xco,\yco){$\kbl$}}
        \xco=158
        \yco=12
        \put(\xco,\yco){$\ktr$}
        \advance\yco by-1
        {\loop
        \put(\xco,\yco){$\kcr$}
        \advance\yco by-2
        \ifnum\yco>-240
        \repeat
        \put(\xco,\yco){$\kbr$}}
        \put(-20,13){\tiny **University of Maryland * Center for String and
         Particle  Theory* Physics Department***University of Maryland *Center
        for String and Particle  Theory** }
        \put(-20,-241.5){\tiny The University of Western Australia * School of Physics **
        The University of Western Australia * School of Physics **
       The University of Western Australia }    
        \end{picture}
        \par\vskip-8mm}
\def\bordero{                                           % alternate border
        \setlength{\unitlength}{1mm}
        \newcount\xco
        \newcount\yco
        \xco=-31
        \yco=12
        \begin{picture}(140,0)
        \put(\xco,\yco){$\ktl$}
        \advance\yco by-1
        {\loop
        \put(\xco,\yco){$\kclr$}
        \advance\yco by-2
        \ifnum\yco>-240
        \repeat
        \put(\xco,\yco){$\kbl$}}
        \xco=151
        \yco=12
        \put(\xco,\yco){$\ktr$}
        \advance\yco by-1
        {\loop
        \put(\xco,\yco){$\kcr$}
        \advance\yco by-2
        \ifnum\yco>-240
        \repeat
        \put(\xco,\yco){$\kbr$}}
        \put(-20,12){\ooo bacdefghidfghghdhededbihdgdfdfhhdheidhdhebaaahjhhdahba

hgdedge
   hgfdiehhgdigicba}
        \put(-20,-241.5){\ooo ababaighefdbfghgeahgdfgafagihdidihiidhiagfedhadbfd

ecdcdfa
   gdcbhaddhbgfchbgfdacfediacbabab}
        \end{picture}
        \par\vskip-8mm}
\def\headpic{                                           % UM heading
        \indent
        \setlength{\unitlength}{.4mm}
        \thinlines
        \par
        \begin{picture}(29,16)
        \put(165,16){\line(1,0){4}}
        \put(170,16){\line(1,0){4}}
        \put(180,16){\line(1,0){4}}
        \put(175,0){\line(1,0){4}}
        \put(180,0){\line(1,0){4}}
        \put(185,0){\line(1,0){4}}
        \put(169,0){\line(0,1){16}}
        \put(170,0){\line(0,1){16}}
        \put(179,0){\line(0,1){16}}
        \put(180,0){\line(0,1){16}}
        \put(184,0){\line(0,1){16}}
        \put(185,0){\line(0,1){16}}
        \put(169,16){\oval(8,32)[bl]}
        \put(170,16){\oval(8,32)[br]}
        \put(179,0){\oval(8,32)[tl]}
        \put(185,0){\oval(8,32)[tr]}
        \end{picture}
        \par\vskip-6.5mm
        \thicklines}
\def\title#1#2#3#4{\border\headpic {\hbox to\hsize{#4 \hfill UMDEPP #3}}\par
        \begin{center} \vglue .5in {\large\bf #1}\\[.6in]
        {#2}\\[.1in] {\it Department of Physics and Astronomy}\\
        {\it University of Maryland, College Park, MD 20742}\\[1.5in]
        {\bf ABSTRACT}\\[.1in] \end{center} \begin{quotation}}  % title stuff
\def\Title#1#2#3#4#5#6#7{\border\headpic
        {\hbox to\hsize{#7 \hfill UMDEPP #6}}\par
        \begin{center} \vglue .4in {\large\bf #1}\\[.4in]
        {#2}\\[.1in] {\it Department of Physics and Astronomy}\\
        {\it University of Maryland, College Park, MD 20742}\\[.1in]
        {#3}\\[.1in] {\it {#4}}\\ {\it {#5}}\\[.5in] {\bf ABSTRACT}\\[.1in]
        \end{center} \begin{quotation}}                 % " for 2 authors
\def\endtitle{\end{quotation}\newpage}                  % end title page

\def\qd{{\kern0.5pt
                   q \kern-5.05pt \raise5.8pt\hbox{$\textstyle.$}\kern
0.5pt}}

\def\dt#1{\on{\hbox{\bf .}}{#1}}                % (big) dot over
\def\Dot#1{\dt{#1}}

\def\gfrac#1#2{\frac {\scriptstyle{#1}}
        {\mbox{\raisebox{-.6ex}{$\scriptstyle{#2}$}}}}
\def\gg{{\hbox{\sc g}}}

\def\dt#1{\on{\hbox{\bf .}}{#1}}                % (big) dot over
\def\Dot#1{\dt{#1}}

\begin{document}

\def\gfrac#1#2{\frac {\scriptstyle{#1}}
        {\mbox{\raisebox{-.6ex}{$\scriptstyle{#2}$}}}}
\def\gg{{\hbox{\sc g}}}
\border\headpic {\hbox to\hsize{September 2010 \hfill
{UMDEPP 10-014}}}
\par
{\hbox to\hsize{$~$ \hfill
{$~$}}}
\par
{$~$ \hfill
{$~$}}
\par

\setlength{\oddsidemargin}{0.3in}
\setlength{\evensidemargin}{-0.3in}
\begin{center}
\vglue .01in
{\large\bf Automorphism Properties of Adinkras
}\\[.3in]

B.\ L.\ Douglas$^\ddag$\footnote{brendan@physics.uwa.edu.au},
S.\, James Gates, Jr.$^\dag$\footnote{gatess@wam.umd.edu},
 and Jingbo B.\ Wang$^\ddag$\footnote{wang@physics.uwa.edu.au}
\\[0.3in]
${}^\dag${\it Center for String and Particle Theory\\[-1mm]
Department of Physics, University of Maryland\\[-1mm]
College Park, MD 20742-4111 USA}
\\[0.1in]
{\it and}
\\[0.1in]
${}^\ddag${\it School of Physics,University of Western Australia,\\[-1mm]
6009, Perth, Australia}\\[0.6in]

{\bf ABSTRACT}\\[.01in]
\end{center}
\begin{quotation}
{Adinkras are a graphical tool for studying off-shell representations of supersymmetry.  
In this paper we efficiently classify the automorphism groups of Adinkras relative to a set of local parameters. Using this, we classify Adinkras according to their equivalence and isomorphism classes. We extend previous results dealing with characterization of Adinkra degeneracy via matrix products, and present algorithms for calculating the automorphism groups of Adinkras and partitioning Adinkras into their isomorphism classes.}

\endtitle

%%%%%%%%%%%%%%%%%%%%%%%%%%%%%%%%%%%%
\section{Introduction}
Several recent studies \cite{Faux05,Doran08,Doran082,Doran083,Gates09,Faux09,Faux092} have introduced and developed a novel approach to the off-shell problem of supersymmetry (see \cite{Doran07} for details). In particular, a graph theoretic tool has emerged to tackle this problem, by encoding representations of supersymmetry into a family of graphs termed Adinkras. This graphical encoding has the advantage of allowing convenient manipulation of these objects, with the goal of achieving a deeper understanding of the underlying representations.

The classification of Adinkras is a natural goal of this work, and various aspects of this have been dealt with extensively in previous studies \cite{Doran08,Doran083,Doran07,Naples09}, in which many of the properties of Adinkra graphs have been ascertained. The works of \cite{Doran08} and \cite{Doran082} relate topological properties of Adinkras to doubly even codes and Clifford algebras.  The 
work in \cite{Naples09} is particularly striking as apparently Betti numbers and functions similar
to those of catastrophe theory seem on the horizon.  However, we will not direct our current
effort toward these observations.

The GAAC (Garden Algebra/Adinkra/Codes) Program began \cite{GR} began with a series of observations of what appeared to be universal matrix algebra structures that seem to occur in
{\em {all}} off-shell supersymmetrical theories.  Two unexpected transformation have occurred
since this start.  First, there emerged Adinkras providing a graphical technology to represent these 
matrix algebras and second the connection of Adinkras to codes.   As the ultimate goal of this program is to provide a definitive classification of {\em {all}} off-shell supersymmetrical theories, 
the appearance of these new unexpected discoveries continue to be encouraging that the goal 
can be reached despite the pessimism surrounding this over thirty year-old unsolved problem.

In this paper, we discuss various definitions of isomorphism of Adinkras. Related to these, we classify the automorphism group of Adinkras in terms of their associated doubly even codes. We characterize this code in terms of local properties of the Adinkras, and hence classify equivalence and isomorphism classes of Adinkras, together with their automorphism group, in terms of a set of efficiently computable local parameters.

The structure of the paper is as follows: Section 2 provides a formal graph theoretic definition of Adinkras, introducing some graph theoretic terms related to their study and discussing some basic properties that are derived from the definition. Section 3 defines the notions of equivalence and isomorphism on the class of Adinkras. We relate Adinkra graphs to doubly even codes, citing a result from \cite{Doran08} that all Adinkras have an associated code. The work of \cite{Doran082} relating properties of Adinkras to Clifford algebras is also discussed, and we define a standard form for Adinkras, used in the proof of later results. Section 4 establishes the main result of this work, classifying the automorphism group of valise Adinkras in terms of the related doubly even code. In Section 5 these results are used to generalize some results of \cite{Gates09}, and provide a polynomial that partitions Adinkras into their equivalence classes. In Section 6 this is extended to isomorphism classes of non-valise Adinkras, and we provide an associated algorithm that accomplishes this partitioning. The appendices detail some of the associated numerical methods and results, and provide some additional examples.

%%%%%%%%%%%%%%%%%%%%%%%%%%%%%%%%%%%%
\section{Adinkra Graphs}

\subsection{Graph-theoretic Notation}

The Adinkra graphs dealt with in this work are simple, undirected, bipartite, edge-N-partite, edge- and vertex-colored graphs. Note that in previous work \cite{Faux05,Doran08,Doran082} Adinkras are considered to be directed graphs. However, as this information is naturally encoded into the height assignment component of the vertex coloring, we remove the edge directions here to simplify the analysis. 

We define a few of the graph theoretic terms below. For a more complete treatment, see \cite{Godsil01}.

A \emph{simple, undirected graph G(V,E)} consists of a vertex set $V$ together with an edge set $E \in V \times V$ of unordered pairs of V.

A graph is \emph{bipartite} if its vertex set can be partitioned into two disjoint sets such that no edges lie wholly within either set. Equivalently, this is a graph containing no odd-length cycles.

A graph $G(V,E)$ is \emph{edge-N-partite} is its edge set $E$ can be partitioned into $N$ disjoint sets, such that every vertex $v \in V$ is incident with exactly one edge from each of these $N$ sets.

Finally, a \emph{coloring} of the edge or vertex set of a graph is a partitioning of these sets into different \emph{color classes}. Formally, this restricts the automorphism group of the graph to the subset that setwise stabilizes these color classes - i.e. the subset that does not map vertices (resp. edges) in one color class to vertices (resp. edges) in another.

%%%%%%%%%%%%%%%%%%%%%%%%%%%%%%%%%%%%
\subsection{Definition and Properties}
\label{sec:def}

Adinkra graphs were introduced in \cite{Faux05} to study off-shell representations of supersymmetry. Thorough definitions are provided in \cite{Doran08,Doran082,Doran07}, although as mentioned above, they vary slightly from the definition presented here, in that for the purposes of this study we will consider them to be undirected graphs.
The most current and complete definition of Adinkras can be found in Definition 3.2 in the work of \cite{Doran08}. 
With that in mind, we use the following definition of Adinrkas for the remainder of this work. An Adinkra graph $G(V,E)$ is a simple, undirected, $N$-regular bipartite graph with the following properties:

The vertices are colored in two ways:
\begin{itemize}
 \item A coloring corresponding naturally to the bipartition (labeled as \emph{bosons} and \emph{fermions}). 
 \item Each edge is given a \emph{height assignment}, hgt: $V \rightarrow \mathbb{Z}$ such that adjacent vertices are at adjacent heights.
\end{itemize}

The edges are also colored in two ways:
\begin{itemize}
 \item A partition into $N$ color classes corresponding to an edge-$N$-partition of the graph. 
 \item An edge parity assignment $\pi : E \rightarrow \mathbb{Z}^2$. We term these two edge types to be \emph{dashed} and \emph{solid}.
\end{itemize}

Furthermore, the connections in the graph are essentially binary in the following manner. Every path of length two having edge colors $(i,j)$ defines a unique 4-cycle with edge colors $(i,j,i,j)$. All such 4-cycles have an odd number of dashed edges. 

We refer to the valence of an Adinkra as its \emph{dimension}. Hence an Adinkra with $N$ edge colors (not counting edge parity) is an $N$-dimensional Adinkra. There will generally be an assumed ordering of the edge colors from 1 to $N$, with the term $i^\textrm{th}$ \emph{edges} or $i^\textrm{th}$ \emph{edge dimension} refering to all edges of the $i^\textrm{th}$ color. The edge parity is often refered to as \emph{dashedness} (see, for example, \cite{Gates09}. In this work we will also refer to it as the \emph{switching state} of a edge or set of edges, motivated by a forthcoming analogy to switching and two-graphs. Relative to the edge color ordering, the switching state of an $i^\textrm{th}$ edge $e = (u,v)$ will be denoted alternately by $\pi_i(u)$, $\pi_i(v)$ or $\pi(e)$ (refering to the $i^\textrm{th}$ edge of the vertex $u$ or $v$, or simply the edge $e$).

The above conditions imply several additional properties. The edge-$N$-partite restriction requires equal numbers of bosons and fermions, hence the bipartition must consist of two sets of equal size. Together with the 4-cycle property, the edge-$N$-partite requirement also implies that the number of vertices, $|V|$, be equal to $2^n$ for some $n \in \mathbb{Z}^+$, $n \le N$. If $n = N$, then ignoring edge and vertex coloring, the resulting graph is simply the $N$-dimensional hypercube, or $N$-cube. Hence we term the corresponding Adinkra the\symbolfootnote[1]{in Section \ref{sec:main} we show that this is unique to $N$} N-\emph{cube Adinkra}. Where $n \ne N$, we denote the corresponding Adinkra to be an \emph{$(N,k)$ Adinkra}, where $k = N-n$.

When drawing these graphs, we will represent the bipartition by black and white vertices. As in previous work \cite{Doran082} the height assignments will be represented by arranging the vertices in rows, incrementally, according to height.

\begin{example}
\label{ex:1}
 The Adinkras drawn below satisfy all requirements listed above. Note that nodes in the same bipartition (bosons and fermions) must be at the same height (modulo 2), and that every 4-cycle containing only two edge colors has an odd number of dashed/solid edges. In case (i), this is simply every 4-cycle, however case (ii) also contains 4-cycles with four edge colors.

 \begin{center}
 $\begin{array}{ccc}
  \includegraphics[width=4.5cm]{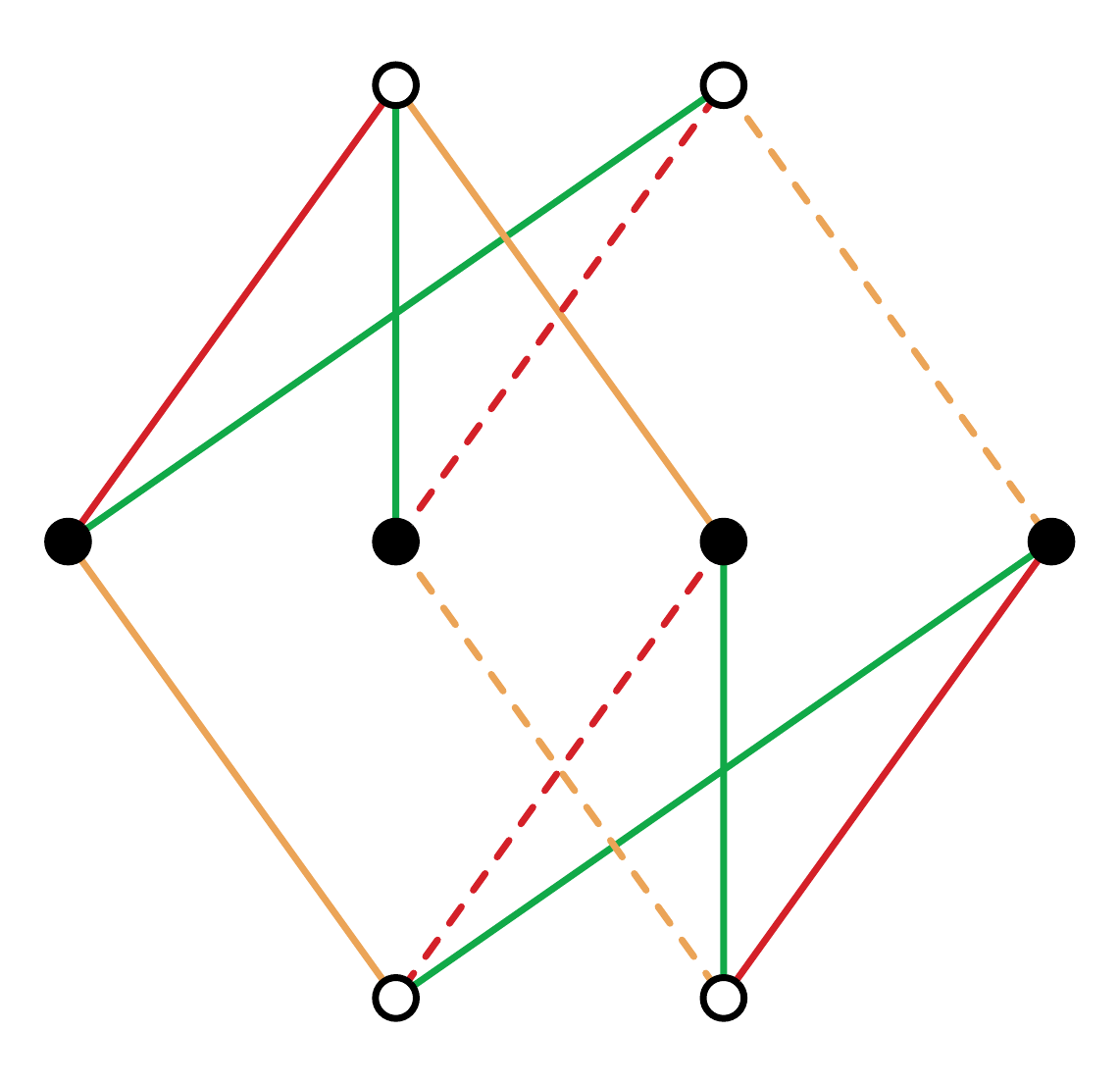} & & \includegraphics[width=4.5cm]{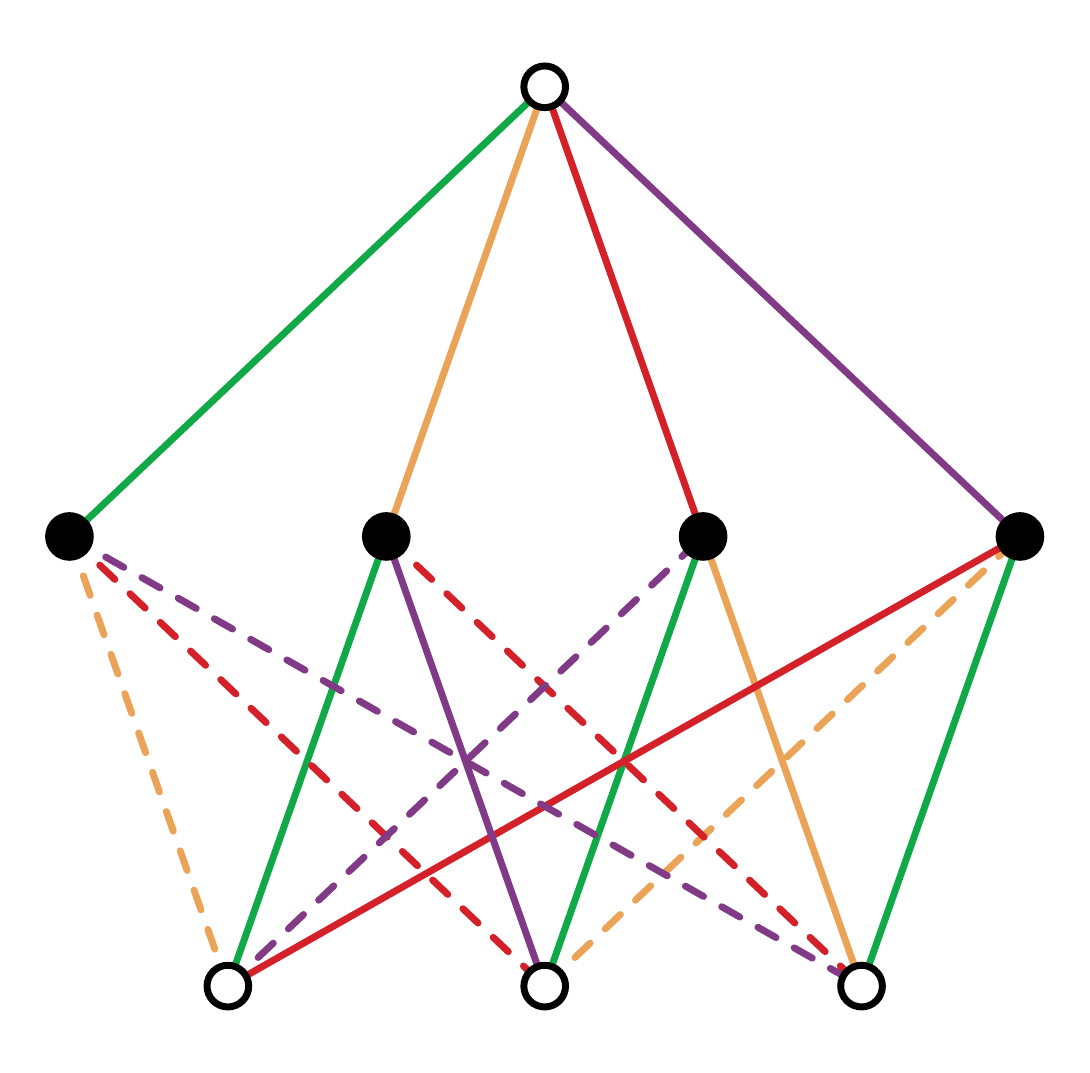} \\
  \textrm{(i) A 3-cube Adinkra.} & & \textrm{(ii) A (4,1) Adinkra.}
 \end{array}$
 \end{center}

\end{example}

It is also worth mentioning at this point that the questions of whether dashed edges correspond to even or odd parity, and which of the black or white node sets correspond to bosons or fermions have been left ambiguous, as they do not impact on the following analysis, and often represent a symmetry in the system. 
Also note that the absolute height value of height assignments are also currently ambiguous, as only relative values of height will be relevant to the following work. Finally, we note that since this work is concerned only with questions of equivalence and isomorphism of Adinkras, the assumption of connectedness will be made throughout to simplify the analysis, without loss of generality.

%%%%%%%%%%%%%%%%%%%%%%%%%%%%%%%%%%%%
\section{Alternative Representations and Models of Adinkras}
\label{sec:notation2}

There are other particularly useful ways of defining and modelling Adinkras. Before describing these, it will be appropriate to introduce some more terminology, specifically relating to linear codes and Clifford algebras. We must also consider how to define notions of equivalence and isomorphism between Adinkras.

\subsection{Equivalence \& Isomorphism Definitions}

\begin{definition}
 The \emph{topology} of an Adinkra is defined as the underlying vertex and edge sets with all the colorings (including edge parity and height assignments) removed. When only the colorings associated with edge parity and height assignments are removed, the resulting graph is termed the \emph{chromotopology} of the original Adinkra.
\end{definition}

Two operations relative to a definition of equivalence have been defined on Adinkras \cite{Doran07}. The \emph{vertex lowering / raising} operation consists of changing the height of a given set of vertices while preserving the requirement that adjacent vertices are at adjacent heights.

The operation of \emph{switching} a vertex consists of reversing the parity of all edges incident to it. We note briefly that this preserves the property that 4-cycles with only two edge colors have an odd number of dashed edges.

This switching operation is analogous to Seidel switching of graphs (see \cite{Brouwer,vanLint66}) in which adjacency and non-adjacency is swapped. In this case, adjacency has been replaced by parity / dashedness for the purposes of switching. Continuing this analogy, we define the \emph{switching class} of an Adinkra $G$ to be the set of all Adinkras that can be obtained from $G$ by switching some subset of its vertices.

Any Adinkras in the same switching class will be considered isomorphic. Adinkras related via vertex raising / lowering operations will be considered equivalent, but not necessarily isomorphic.

Hence the definitions of equivalence and isomorphims considered here differ slightly. Permuting the vertex labels, switching and raising / lowering operations all preserve equivalence, whereas the only operations preserving isomorphism are those of switching and permutations of vertex labels.

\begin{example}
 All three Adinkras drawn below are in the same equivalence class, however only the first two are isomorphic.
 \begin{center}
 $\begin{array}{ccccc}
  \includegraphics[width=4.5cm]{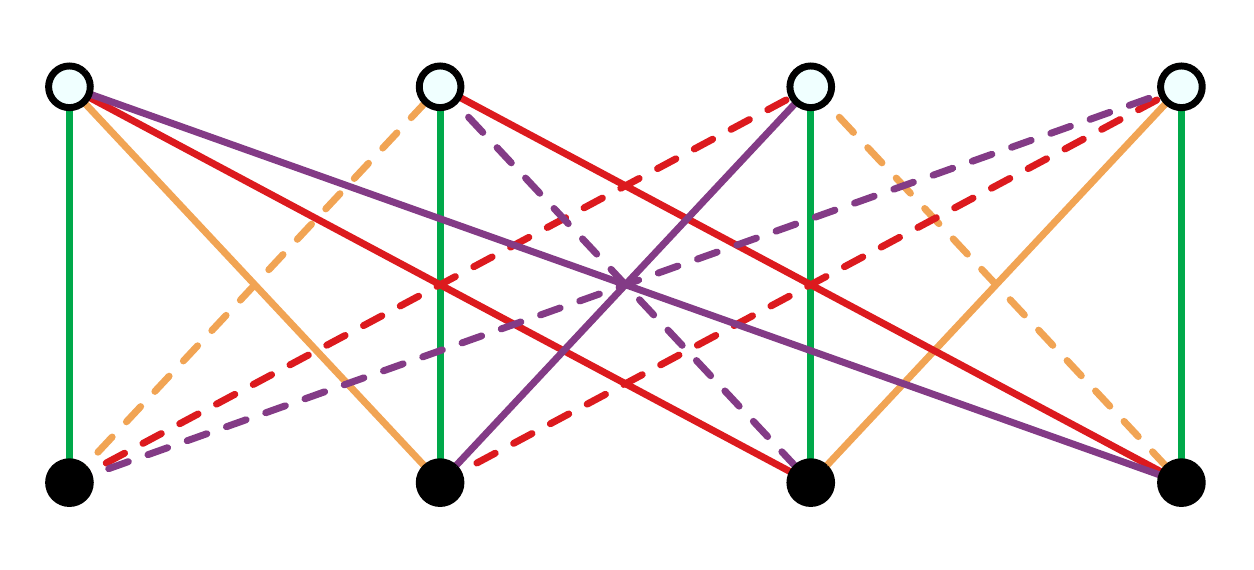} & & \includegraphics[width=4.5cm]{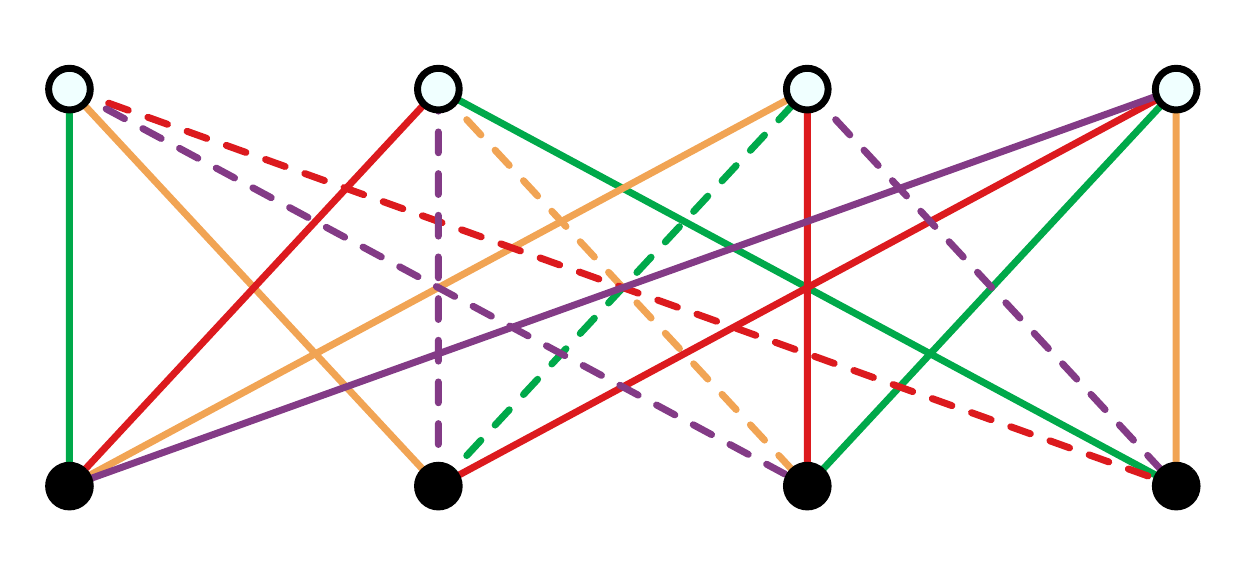} & & \includegraphics[width=4.5cm]{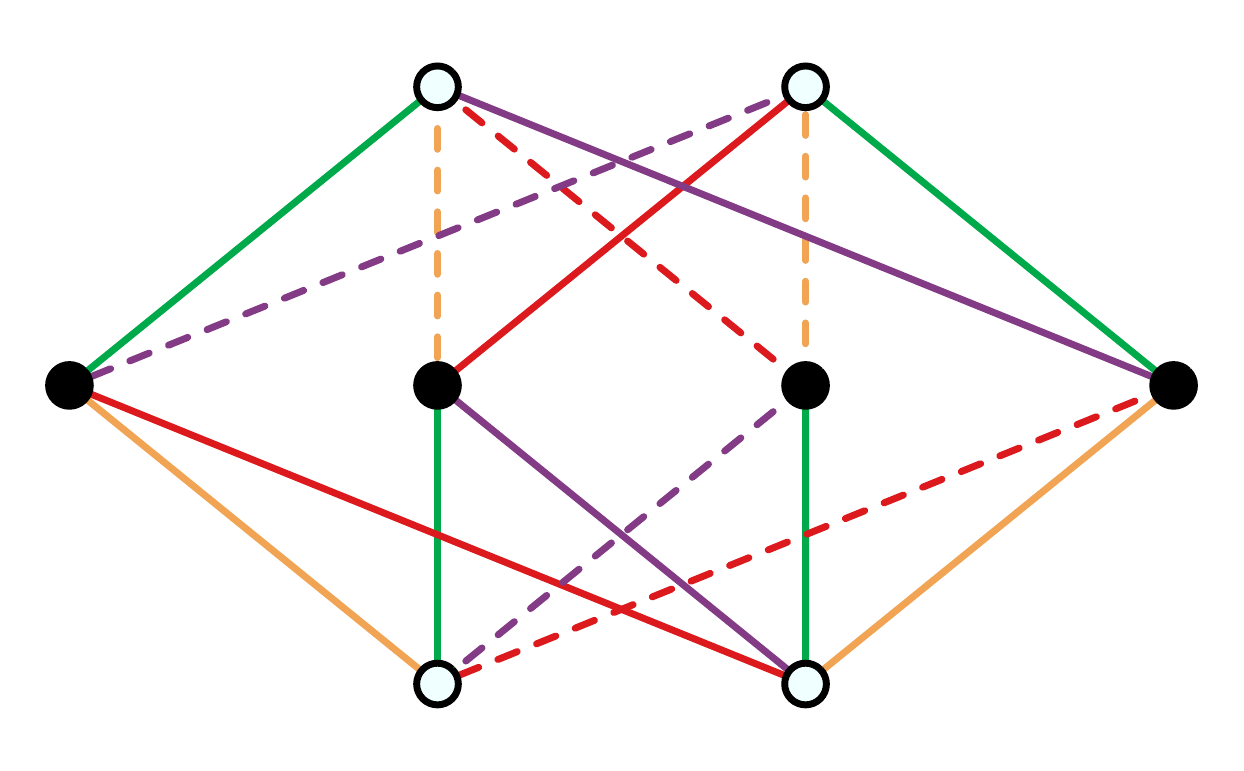} %\\
 \end{array}$
 \end{center}

\end{example}

An \emph{automorphism} of an $(N,k)$ Adinkra $G(V,E)$ is a permutation of the vertex set, $p : V \rightarrow V$, $p \in \textrm{Sym}(2^{N-k})$ that is also an isomorphism. To introduce the tools necessary to classify the automorphism group properties of Adinkras, it will first be necessary to give a brief description of linear codes and Clifford algebras.

\subsection{Linear Codes}

A binary linear $(N,k)$-code $C$ consists of a set of $2^k$ codewords (bit-strings) of length $N$ forming a group under their bitwise sum modulo 2 (exclusive or), represented by $\boxplus$.
Hence for all $u,v \in C$, $u \boxplus v \in C$. When dealing with a bit-string or codeword $u$ of length $N$ we will always use the notation $u = (u_1,u_2,\ldots,u_N)$, where $u_i$ is the $i^\textrm{th}$ bit of $u$. We will also consider all subsequent codes to be both binary and linear; these terms are henceforth dropped from their description.

A \emph{generating set} of $C$ is a set of $k$ codewords $D = (c_1,\ldots, c_k) \subset C$ such that all codewords in $C$ can be formed by a linear combination of elements of $D$. The value $k$ is called the \emph{dimension} of the code, and is invariant with respect to the particular choice of generating set.

The \emph{weight} of a codeword $u$, denoted by wt($u$), or simply $|u|$, is the number of 1's in $u$. If every codeword in a code has a weight of 0 (mod 4) the code is called \emph{doubly even}. The inner product of two codewords $u$ and $v$ of length $N$ is defined as 
\be
\langle u,v\rangle \equiv \sum_{i=1}^N u_i\,v_i \pmod{2}. 
\ee
Then a doubly even code has the property that any two of its codewords have an inner product of 0.

\begin{definition}
 A \emph{standard form} generating set $D$ of an $(N,k)$-code $C$ is defined to be of the form $D = ( I \; | \; A )$, where $I$ is the $k \times k$ identity matrix and $A$ is a $k \times (N-k)$ matrix being the remainder of each codeword in $D$.
\end{definition}

\begin{example}
By considering a permutation of the bit-string, and an alternative generating set, the code $d_8$, generated by
$\begin{pmatrix}
 1\,1\,1\,1\,0\,0\,0\,0\\[-1mm]
 0\,0\,1\,1\,1\,1\,0\,0\\[-1mm]
 0\,0\,0\,0\,1\,1\,1\,1\\[-1mm]
\end{pmatrix}$
can be written in standard form
$\left(\hspace{-0.2cm}\begin{array}{ccc|ccccc}
 1\,0\,0\,1\,1\,1\,0\,0\\[-1mm]
 0\,1\,0\,1\,1\,0\,1\,0\\[-1mm]
 0\,0\,1\,1\,1\,0\,0\,1\\[-1mm]
\end{array}\hspace{-0.2cm}\right)$.
Note that generating sets will be assumed to be minimal, in that the codewords comprising the set are linearly independent. In this case, a generating set with $n$ codewords corresponds to a code with $2^n$ distinct codewords.
\end{example}

\begin{observation}
  For any doubly even $(N,k)$-code, there exists a generating set $D$ and permutation of the bit-string order $p \in \textrm{Sym}(N)$ such that $D$ is a standard form generating set.
\end{observation}

The notion of bit-strings can be effectively used to furnish a further organization of the edge and vertex sets of an Adinkra. If we order the edge colors of an $N$-cube Adinkra from 1 to $N$, then assign each of these edge colors to the corresponding position in an $N$-length bit-string, each of the $2^N$ vertices in the Adinkra can be represented by such a bit-string. These bit-strings are allocated in such a way that two vertices connected by the $i^\textrm{th}$ edge differ precisely in the $i^\textrm{th}$ element of their respective bit-strings.

\begin{example}
\label{ex:ad2}
 Applying this bit-string representation to the vertices of a 3-cube Adinkra yields the graph below. Note that for $N$-cube Adinkra, vertices are connected if and only if their corresponding bit-string representations have a Hamming distance of 1 (differ in only one position).
 
 \begin{center}
  \includegraphics[width=4.5cm]{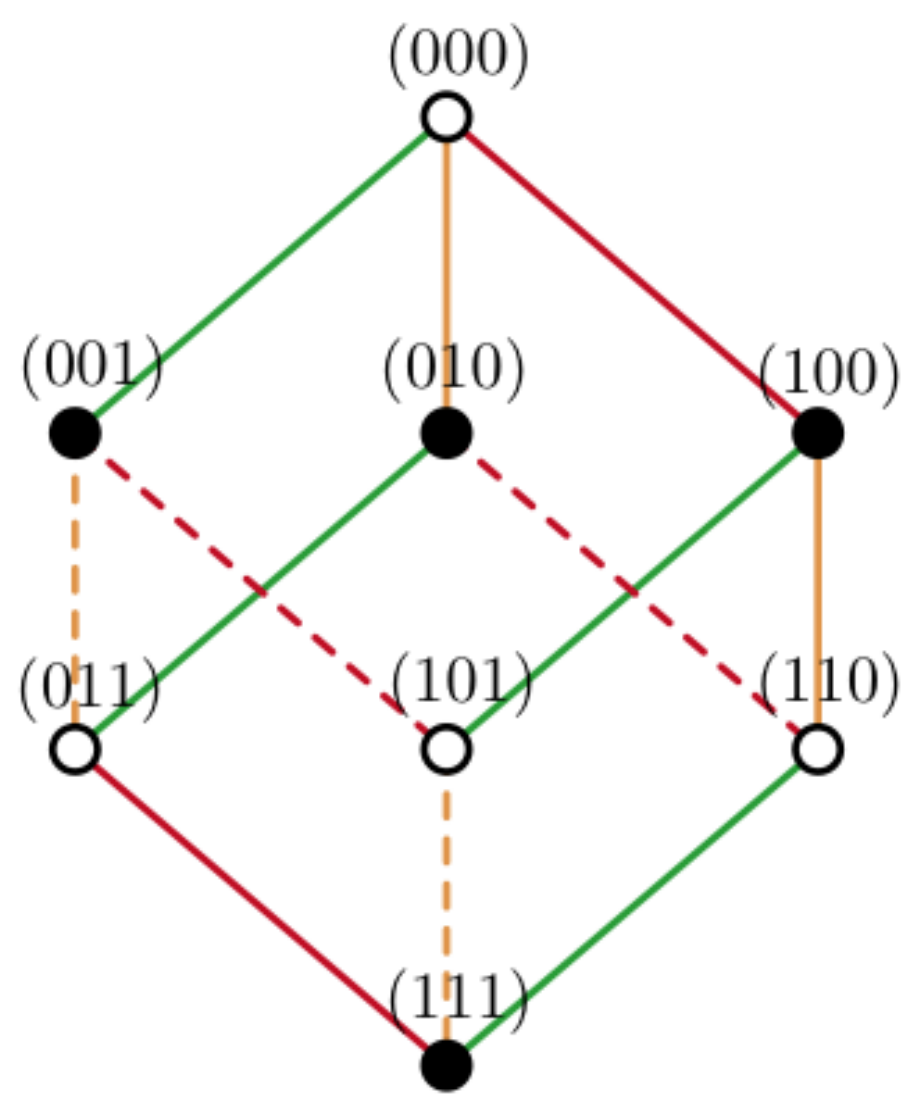}
 \end{center}

\end{example}

An important theorem of \cite{Doran08} links general $(N,k)$ Adinkras to double even codes. We paraphrase it, combined with other results from the same work, below.

\begin{theorem}[\cite{Doran08}]
\label{thm:quotient}
 Every $(N,k)$ Adinkra (up to equivalence) class can be formed from the $N$-cube Adinkra by the process of quotienting via some doubly even $(N,k)$ code.
\end{theorem}

This quotienting process works as follows. We start with a doubly even $(N,k)$ code $C$ and an $N$-cube Adinkra with a corresponding edge color ordering and bit-string representation. We then identify vertices (and their corresponding edges) related via any codeword in $C$, ensuring first that identified vertices have identical edge parities relative to the ordering of edge colors. Thus an $(N,k)$ code identifies groups of $2^k$ vertices, and reduces the order of the vertex set from $2^N$ to $2^{N-k}$, producing an $(N,k)$ Adinkra.

\subsection{Clifford Algebras}

\begin{definition}
 We will consider the \emph{Clifford algebra} $Cl(n)$ to be the algebra over $\mathbb{Z}_2$ with $n$ multiplicative generators $\gamma_1,\gamma_2,\ldots,\gamma_n$, with the property
 \be
  \{ \gamma_i, \gamma_j \} = 2 \delta_{ij} \mathbf{1}, 
  \ee
 where $\delta_{ij}$ is the Kroenecker delta. In other words, each of the generators is a root of 1, and any two generators anticommute.
\end{definition}

It will be convenient to view the 4-cycle condition of Adinkras in terms of the anticommutativity property of Clifford generators. If we consider each position in the bit-string representation of a vertex of an $N$-cube Adinkra to correspond to a given generator of $Cl(N)$, with the vertex itself being related to the product of these Clifford generators, then the $i^\textrm{th}$ edge dimension can be naturally viewed as the transformations corresponding to left Clifford multiplication by the $i^\textrm{th}$ Clifford generator.

For example, a vertex with bit-string representation $(01100)$ corresponds to the product of Clifford generators $\gamma_2 \gamma_3 \in Cl(5)$, and the edge connecting the vertex $(01100)$ to $(01101)$ is related to the mapping $\gamma_2 \gamma_3 \mapsto (\gamma_5)(\gamma_2 \gamma_3) = \gamma_2 \gamma_3 \gamma_5$. Then the condition that every 4-cycle containing exactly two edge colors has an odd number of dashed edges corresponds to the anticommutativity property of Clifford generators, $\gamma_i \gamma_j = -\gamma_j \gamma_i$.

In connection with this Clifford generator notation, we define a standard form for the switching state of an $N$-cube Adinkra.

\begin{definition}
\label{def:stdcube}
 A \emph{standard form} $N$-cube Adinkra has edge parity as follows. The $i^\textrm{th}$ edge of vertex $v = (v_1,v_2,\ldots,v_N)$ has edge parity given by
\be
 \pi_i(v) \equiv (v_1 + v_2 + \ldots + v_{i-1}) \pmod{2}. 
 \ee
 This corresponds to left Clifford multiplication of $v$ by $\gamma_i$. Note that this standard form is defined relative to some given bit-string labeling of the vertex set. Since fixing the labeling of a single vertex fixes that of all vertices, we will sometimes refer to a standard form as being relative to a \emph{source node}, being the vertex with label $(0,0,\ldots,0)$, or simply standard form relative to 0.
\end{definition}

\begin{example}
\label{ex:3}
 The 3-cube Adinkra of Example \ref{ex:ad2} is in standard form. Vertices are ordered into heights relative to their bit-strings, such that heights range from 0 (at the bottom) to 3, and vertices at height $i$ have weight $i$. The edges are ordered from left to right, such that green corresponds to $\gamma_1$ and red corresponds to $\gamma_3$.
 
\end{example}

As we have seen, many\footnote{As we have not yet investigated the case of {\em {gnomon Adinkras}} defined by the `zippering' process 
\newline $~~~\,~~$
presented in {\cite{X}}, extension of our results to these require further study.} Adinkras can be formed by quotienting the $N$-cube Adinkra with respect to some doubly even code. However in practice this method can be quite inefficient. The following alternative method for constructing an $(N,k)$ Adinkra, with associated code $C$ is due to G. Landweber (personal communication), and is used by the Adinkramat software package:

\begin{enumerate}
 \item Start with the standard form $(N-k)$-cube Adinkra induced on the first $N-k$ edge dimensions.
 \item Find a standard form generating set for $C$, of the form $D = ( I \; | \; A )$.
 \item Associate the $(N-k+i)^\textrm{th}$ edge dimension with the product of Clifford generators given by $A_i$, the $i^\textrm{th}$ row of A.
 \item The $i^\textrm{th}$ edge connects a vertex $v$ to the vertex $v A_i$, with switching state corresponding to right Clifford multiplication of $v$ by $A_i$, up to a factor of $(-1)^F$.
\end{enumerate}

The factor of $(-1)^F$, termed the fermion number operator, simply applies a $(-1)$ factor to fermions, and leaves bosons unchanged. This factor is required to ensure the parity of an edge remains the same in either direction. Again, note that the choice of whether even- or odd-weight vertices are bosons remains ambiguous. For simplicity, we will consider odd-weight vertices to be fermions for the purposes of the $(-1)^F$ factor, with the symmetry encoded by a graph-wide factor of $\pm 1$ in the switching state of the extra $k$ edge dimensions, corresponding to two potentially inequivalent $(N,k)$ Adinkras.

It remains to be seen whether this construction method yields all $(N,k)$ Adinkras. To our knowledge this has not been explicitly proven in previous work, hence we present a proof of this as follows.
\begin{theorem}
\label{thm:construction}
 All $(N,k)$ Adinkras of a given equivalence class can be formed from the $(N-k)$-cube Adinkra by the construction method detailed above.
\end{theorem}

\begin{proof}
As we are only considering equivalence classes, the height assignments will be ignored.

First consider an $N$-cube Adinkra quotiented with respect to an $(N,k)$ code with standard form generating set $D = ( I \; | \; A )$, producing an $(N,k)$ Adinkra $G(V,E)$. Assume without loss of generality that the vertices remaining after quotienting all have fixed $(N-k+1)^\textrm{th} \rightarrow N^\textrm{th}$ bit-string characters (for instance all fixed as 0). Then the rows of $D$ do indeed represent the vertices identified by the quotienting process. We want to show that, up to isomorphism, there are at most two different ways this quotienting operation can be performed.

Choosing only $N-k$ of the edge dimensions of $G$ yields an induced $(N-k)$-cube Adinkra, $H(V,E')$, with $E' \subset E$. We will see in Section \ref{sec:main} that for fixed $m$, all $m$-cube Adinkras belong to a single equivalence class. Hence without loss of generality $H$ can be considered to be in standard form. This in turn fixes the switching state of all edges belonging to the chosen $N-k$ edge dimensions. Now only a single degree of freedom remains in the switching states of the remaining $k$ edge dimensions, due to the 4-cycle condition. In other words, fixing the switching state of any one of the remaining edges fixes all remaining edges. This choice corresponds to a graph-wide factor of $\pm 1$ in the switching state of the additional $k$ edge dimensions.

Since this graph-wide factor of $\pm 1$ matches the difference between the two possibly inequivalent Adinkras produced by the above construction method, it remains to show that the construction method does indeed produce valid $(N,k)$ Adinkras. Hence me must verify that the 4-cycle condition holds.

As odd-weight vertices are considered to be fermions, the $(-1)^F$ factor can be replaced by an additive factor of $\sum_{j=1}^{N-k} x_j$, for a given vertex $x \in V$. Then for $i > N-k$, $x \in V$, $A_i = (a_1,a_2,\ldots,a_{N-k})$, 
\begin{align}
 \nonumber \pi_i(x) &\equiv \sum_{j=1}^{N-k} x_j \;+\; x \ .\  A_i \\
 &\equiv \sum_{j=1}^{N-k} x_j \;+\; x_2 (a_1) + x_3 (a_1 + a_2) + \ldots + x_{N-k} (a_1 + a_2 + \ldots + a_{(N-k-1)})\pmod{2}.
\end{align}

The 4-cycles in $G$ with two edge colors $i$ and $j$ will be split into three cases:
\begin{itemize}
 \item[(i)] $i,j \le N-k, i < j$
 \item[(ii)] $i \le N-k, j > N-k$ and
 \item[(iii)] $i,j > N-k$.
\end{itemize}
Note that a given vertex $x \in V$ defines a unique such 4-cycle. In case (i), consider two antipodal points of such a 4-cycle, $x$ and $x+i+j$. Then the sum of the edge parities of the 4-cycle equals:
\begin{align}
 &(x_1 + \ldots + x_{i-1}) \;+\; (x_1 + \ldots + x_{j-1}) \;+\;  \nonumber  \\
 &(x_1 + \ldots + x_{i-1}) \;+\; (x_1 + \ldots +\; (x_j + 1) \;+ \ldots + x_{j-1}) \equiv 1 \pmod{2}.
\end{align}

Hence the 4-cycle condition holds (note that this is the only case present in the construction of the standard form $N$-cube Adinkra).

For case (ii), the antipodal points are $x$ and $z$, where
\begin{align}
 z \equiv & ~(x_1 + a_1, \ldots, x_i + a_i + 1, \ldots, x_{N-k} + a_{N-k})   \nonumber  \\
 \equiv & ~x + A_j + i \pmod{2},
\end{align}
and by substituting the equation of (1) above, the sum becomes:
\begin{align}
 &(x_1 + \ldots + x_{i-1}) \;+\; (\sum_{l=1}^{N-k} x_l + x . A_j) \;+\; (x_1 + \ldots + x_{i-1}) \;+\; \nonumber \\
 &(a_1 + \ldots + a_{i-1}) + (\sum_{l=1}^{N-k} z_l + z . A_j) \pmod{2}.
\end{align}
However $|A_j| \equiv 3 \pmod{4}$, so $|x| \equiv |z| \pmod{2}$. Then the first and third terms cancel, as do the second and fifth, leaving:
\begin{align}
 x.A \;+\; z.A \;+\;& (a_1 + \ldots + a_{i-1}) \nonumber \\
 &\equiv 2(x.A) + 2(a_1 + \ldots + a_{i-1}) + a_2 (a_1) + a_3 (a_1 + a_2)  \nonumber \\
 & {~~~~} + \ldots + a_{N-k} (a_1 + \ldots a_{(N-k-1)} )
  \nonumber \\
 &\equiv |A_j -1| + |A_j - 2| + \ldots + 1 \nonumber \\
 &\equiv \frac{1}{2} |A-1||A| \pmod{2}.
\end{align}
Now $|A| \equiv 3 \pmod{4}$, hence
\begin{align}
 \frac{1}{2} |A-1||A| &\equiv 3 \pmod{4} \nonumber \\
 &\equiv 1 \pmod{2},
\end{align}
and case (ii) is verified.

Case (iii) proceeds similarly to case (ii).

Hence the graphs produced via this construction method are indeed $(N,k)$ Adinkras. Since \cite{Doran08} showed that all $(N,k)$ Adinkras can be produced via the quotienting method, there can be at most two equivalences classes of Adinkras with the same associated code, and hence the Adinkras produced from the two methods must coincide.

\end{proof}

\begin{definition}
\label{def:std form}
 A \emph{standard form} $(N,k)$ Adinkra has switching state given by the above construction, relative to an associated doubly even code.
\end{definition}

\begin{example}
\label{ex:N=4}
 The smallest non-trivial $(N,k)$ Adinkra has parameters $N=4, k=1$, with associated doubly even code $(1111)$. The \emph{standard form} described in definition \ref{def:std form} is shown below, for each choice of the graph-wide factor of $\pm 1$.

 \begin{center}
  \includegraphics[width=6.5cm]{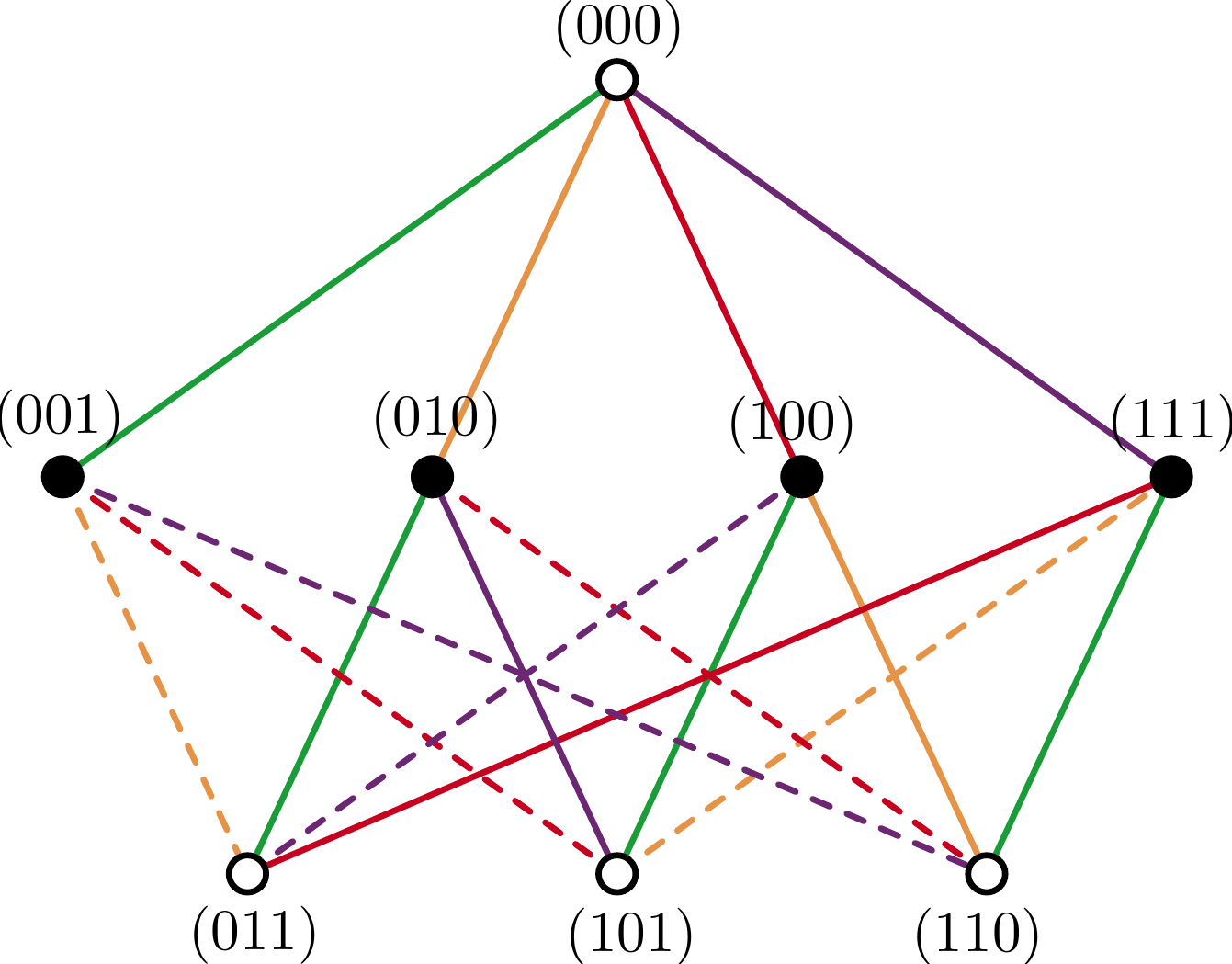} \hspace{1.5cm} \includegraphics[width=6.5cm]{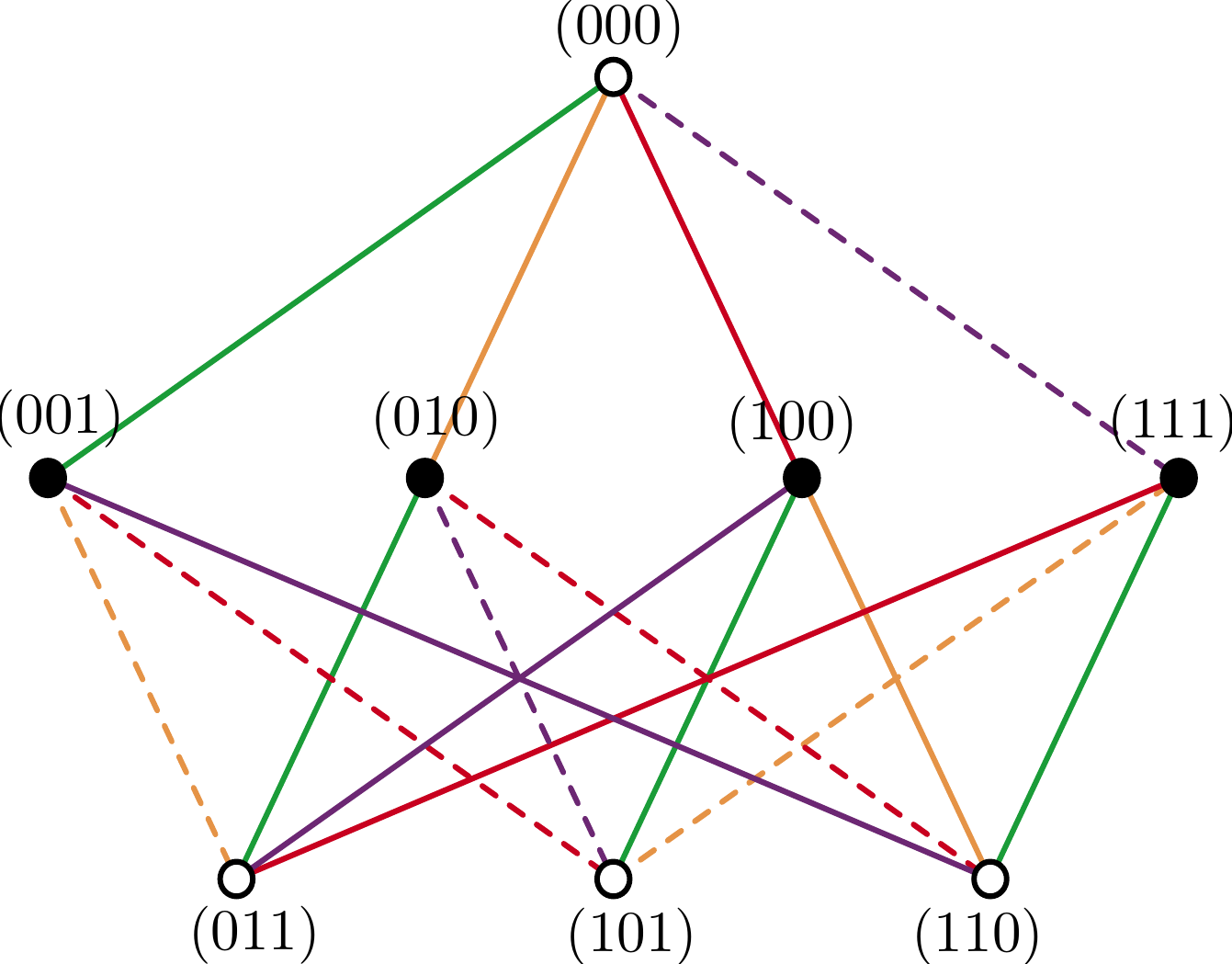}
 \end{center}
 In this case the two Adinkras formed belong to different equivalence classes, as we will demonstrate in the following section.

\end{example}

This construction provides a much simpler practical method of constructing Adinkras, and the standard form will prove convenient to prove later results. Note that all $(N,k)$ Adinkras have an associated doubly even code, with matching parameters of length and dimension. For these purposes the $N$-cube Adinkra is considered an $(N,0)$ Adinkra with a trivial associated code of dimension 0.

\setcounter{equation}{0}
%%%%%%%%%%%%%%%%%%%%%%%%%%%%%%%%%%%%
\section{Automorphism Group Properties of Adinkras}
\label{sec:main}

We now have sufficient tools to derive the main result of this paper: classifying the automorphism group of an Adinkra with respect to the local parameters of the graph, namely the associated doubly even code. We start with several observations regarding the automorphism properties of Adinkras.

\begin{lemma}
\label{thm:unique}
 Ignoring height assignments\symbolfootnote[2]{Equivalently, there are only two heights, corresponding to bosons and fermions. Adinkras 
with \newline $~~~\,~~$ 
this property are called \emph{valise} Adinkras. Adinkras with more than two heights are termed
\newline $~~~\,~~$
 \emph{non-valise}.}, the $N$-cube Adinkra is unique to $N$, up to isomorphism.
\end{lemma}

\begin{proof}
 It suffices to show that any $N$-cube Adinkra can be mapped to standard form via some set of vertex switching operations. Such a mapping can be trivially constructed. For example, start with a single vertex. The parity of its edge set can be mapped to any desired form by a set of vertex switchings applied \emph{only} to its neighbors. In particular the edge parities can be mapped to those of the corresponding standard form Adinkra. Since the 4-cycle restriction of Section \ref{sec:def} holds, this process can be continued consistently for set of vertices at each subsequent distance from this original vertex, regardless of the initial switching state of the Adinkra.
\end{proof}

\begin{lemma}
\label{thm:vt}
 $N$-cube valise Adinkras are vertex-transitive up to the boson/fermion bipartition. 
\end{lemma}

\begin{proof}
To show this, we begin with a standard form $N$-cube Adinkra $G(V,E)$, with switching state defined by left multiplication by the corresponding Clifford generators, as in definition \ref{def:stdcube}. Since $N$-cube Adinkras are unique up to isomorphism, in the sense of lemma \ref{thm:unique}, this can be done without loss of generality.

Then the switching state of edge $(x,y)$, where $x = i.y$, is of the form $x_1 + x_2 + \ldots + x_{i-1} \pmod{2}$. Consider a non-trivial permutation $\rho : V \mapsto V$. Since $\rho$ is non-trivial, there are two vertices $u,v$ in V, $u \ne v$, such that $\rho : u \mapsto v$. If $\rho$ is to preserve the topology of $G$, we must have $\rho : x \mapsto (u \; v) x, \, \forall \, x \in V$. Note that $\rho$ has no fixed points, and is completely defined in terms of $u.v$. Any automorphism $\gamma$ of $G$ which permutes $V$ according to $\rho$ must also switch some set of vertices in such a way that edge parity is preserved. However $\rho$ maps the edge $(x,y)$, where $x = i.y$, with edge parity $x_1 + x_2 + \ldots + x_{i-1} \pmod{2}$ to the edge $(\rho(x),\rho(y))$, with edge parity of
\be
(x_1 + x_2 + \ldots + x_{i-1}) + ((u_1+v_1) + (u_2+v+2) + \ldots + (u_{i-1} + v_{i-1})) \pmod{2}. 
\ee
Hence the \emph{change} in edge parity of $(x,y)$ does not depend on either endpoint explicitly, only on the mapping $\rho$ and the edge dimension $i$. In other words, $\rho$ either preserves the switching state of all edges of a given edge dimension, or reverses the parity of all such edges. Note that any single edge dimension can be switched (while leaving all other edges unchanged) via a set of vertex switching operations, hence an automorphism $\gamma$ exists for all such $\rho$. Furthermore, $\gamma$ is unique to $\rho$, which is in turn unique to the choice of $(u \; v)$.

Note that if $G$ is in standard form, we term the automorphism mapping vertex $x$ to vertex $0 = (00 \ldots 0)$ to put $G$ in standard form \emph{relative to $x$}. By this terminology, $G$ was initially in standard form \emph{relative to 0}. 

\end{proof}

\begin{corollary}
 The $N$-cube Adinkra is minimally vertex-transitive (up to the bipartition), in the sense that the pointwise stabilizer of the automorphism group is the identity, and $|\textrm{Aut}(G)| = \frac{1}{2}|V| = 2^{N-1}$. In other words, no automorphisms exist that fix any points of the $N$-cube Adinkra.
\end{corollary}

Consider any $(N,k)$ Adinkra $G(V,E)$. The sub-Adinrka $H(V,E')$ induced on any set $S$ of $(N-k)$ edge dimensions of $G$ will be an $(N-k)$-cube Adinkra. Also, for all such $S$ there exists some set of vertex switching operations such that the induced Adinkra $H$ is in standard form. In the case where $S = (1,2,\ldots,N-k)$, the induced $H$ is in standard form if and only if $G$ is in standard form.

\begin{observation}
\label{obs1}
 Given some $(N,k)$ Adinkra, containing a standard form $(N-k)$-cube Adinkra induced on the first $(N-k)$ edge dimensions (or equivalently an induced $(N-k)$-cube Adinkra in \emph{any} given form), the set of edge parities of the $i^\textrm{th}$ edge dimension, $i > (N-k)$, are fixed up to a graph wide factor of $-1$. 
\end{observation}

In particular, fixing the parity of the \emph{extra} edges (where $i > (N-k)$) of a single vertex fixes the parity of all extra edges in the graph. This is a direct consequence of the anticommutativity property of 4-cycles with two edge colors.

This reduces the problem of finding automorphisms of a general $(N,k)$ Adinkra to that of local mappings between vertices. In particular, for any two nodes $x,y \in V$, there will be a unique automorphism of the induced $(N-k)$-cube Adinkra mapping $x$ to $y$. This will extend to a full automorphism of $G$ if and only if it preserves the switching state of the additional $k$ edges of $x$. Hence we arrive at the following result.

\begin{theorem}
\label{thm:main}
 Consider an $(N,k)$ Adinkra $G(V,E)$ having an associated code $C$ with standard form generating set $D = (I \; | \; A )$. Two vertices $x,y$ in $V$ are equivalent (disregarding vertex coloring, there is an automorphism mapping between them) if and only if their relative inner products with respect to each codeword in $D$ are equal. In other words, $\exists \; \gamma \in \textrm{ Aut(G) such that } \gamma : x \mapsto y \textrm{ iff } \forall \; c \in D, \langle x,c\rangle \equiv \langle y,c\rangle \pmod{2} $.
\end{theorem}

\begin{proof}
 Consider without loss of generality the case where $G$ is in standard form, and take $H(V,E')$ to be the sub-Adinkra induced on the first $N-k$ edge dimensions. By lemma \ref{thm:vt}, $H$ is vertex transitive if we disregard vertex coloring. So consider the automorphism $\rho$ of $H$, $\rho : x \mapsto y$, where $x,y \in V$. Now $\rho$ consists of two parts. A permutation of the vertex set, $\rho : v \mapsto x \; y \; v, \forall \;v \in V$, and a set of vertex switching operations, such that the switching state of $H$ is preserved.

 We wish to know when $\rho$ switches $x$ relative to $y$ (i.e. when $x$ is switched but $y$ is not, or vice versa, and when either both or neither are switched). Where $\pi_i(x)$ denotes the switching state of the $i^\textrm{th}$ edge of $x$ (in $G$ or $H$), we will use $\pi'_i(x)$ to denote the equivalent switching state of $x$ in $\rho(G)$ or $\rho(H)$ (we will also denote $\rho(G)$ and $\rho(H)$ by $G'$ and $H'$ respectively. Note that $x$ in $G'$ is the image of $y$ in $G$ under $\rho$.

 For edge dimension $i \le N-k$, we see that since $\rho$ is an automorphism of $H$, $\pi_i(x) = \pi'_i(y)$, and similarly $\pi_i(v \; x) = \pi'_i(v \; y) \;\forall \;v \in V$. For the case $i > N-k$, $\pi_i(v)$ is given by equation (1) in the proof of theorem \ref{thm:construction}, but what about $\pi'_i(v)$?

 The value of $\pi'_i(v)$ can be found by considering the path $P_i(x)$ in $H$ joining $x$ to $x.A_i$, where $A_i = (a_1,a_2,\ldots,a_{N-k})$. If we denote $A_i$ by the associated product of Clifford generators, $A_i = \prod_{j=1}^{n} \gamma_{s_j}$, where the $s_j$ denote the 1's of $A_i$, we have $P_i(x) = (x,x.a_{s_1},(x.a_{s_1}).a_{s_2},\ldots,x.A_i)$. Then $\rho$ will preserve the parity of edge $(x,x.A_i)$ if and only if the number of dashed edges in paths $P_i(x)$ and $P_i(y)$ are equal (mod 2). In other words,
\be
 \pi_i(x) = \pi'_i(x) \; \textrm{ iff } \sum_{e \in P_i(x)}\pi(e) \equiv \sum_{e \in P_i(y)}\pi(e).
 \ee

 Now $\pi_{s_1}(x) \equiv x_1 + \ldots + x_{s_1-1} \pmod{2}, \ldots, \pi_{s_n}(x.A_i) \equiv (x_1 a_1) + \ldots + (x_{s_n} a_{s_n}) \pmod{2}$, where $a_i = 1$ for $i \in \{s_j : 1 \le j \le n\}$, and $a_i = 0$ elsewhere. Hence this can be rearranged as:
 \begin{align}
 \sum_{e \in P_i(x)}\pi(e) &\equiv a_2 (x_1) + a_3 (x_1 + x_2) + \ldots + a_{N-k} (x_1 + \ldots + x_{N-k-1})\nonumber \\
  &\equiv x_1 (a_2 + \ldots + a_{N-k}) + \ldots + x_{N-k-1} (a_{N-k}) \pmod{2}
 \end{align}

 Hence, $\pi_i(x) \equiv |x| + x.A_i \pmod{2}$, and $\pi'_i(y) \equiv \pi_i(x) + (\sum_{e \in P_i(x)}\pi(e) - \sum_{e \in P_i(y)}\pi(e))$, and note that $\sum_{e \in P_i(x)}\pi(e) + x.A_i$ simplifies to
 \be
  |A_i||x| - \sum_{j=1}^{N-k} (x_j a_j), 
  \ee
 so we have, for $i > N-k$,
 \begin{align}
  \pi_i(x) - \pi'_i(y) \equiv |x| + |A_i||x| + |y| + |A_i||y| + \sum_{j=1}^{N-k} (x_j a_j) + \sum_{j=1}^{N-k} (y_j a_j) \pmod{2}.
 \end{align}
 Since $|A_i| \equiv 1 \pmod{2}$, we have $|x|(|A_i| + 1) \equiv 0 \pmod{2}$, and the first 4 terms cancel, leaving
 \begin{align}
  \pi_i(x) - \pi'_i(y) &\equiv \sum_{j=1}^{N-k} (x_j a_j) + \sum_{j=1}^{N-k} (y_j a_j)  \nonumber \\
  &\equiv \langle x,A_i\rangle + \langle y,A_i\rangle \pmod{2},
 \end{align}
 which completes the proof.

\end{proof}

\begin{observation}
 Two vertices having the same relative inner products with respect to each element of a set $S$ of codewords also have the same relative inner product with respect to the group generated by $S$, under bitwise addition modulo 2.
\end{observation}

\begin{corollary}
 Conversely to theorem \ref{thm:main}, consider a single orbit $\phi$ of the $(N,k)$ Adinkra $G$ with associated doubly even code $C$. If all elements of $\phi$ have fixed inner product relative to an $N$-length codeword $c$, then $c \in C$, the code by which $G$ has been quotiented.
\end{corollary}

Theorem \ref{thm:main} leads to several important corollaries regarding the automorphism group properties of Adinkras. Recall that there are at most two equivalence classes of Adinkras with the same associated doubly even code. Theorem \ref{thm:main} implies that disregarding the vertex colorings there is in fact only one such equivalence class. Hence including the vertex bipartition, we see that two equivalence classes exist if and only if the respective vertex sets of bosons and fermions are setwise non-isomorphic. This in turn occurs if and only if at least one orbit (and hence all orbits) of $G$ is of fixed weight modulo 2.

In \cite{Doran08} and \cite{Gates09}, one such case was investigated, for Adinkras with parameters $N = 4, k = 1$. In this case, the Klein flip operation, which exchanges bosons for fermions, was found to change the equivalence class of the resulting Adinkra\footnote{This fact was well known to the authors of
\cite{Doran08}. }. We term this property Klein flip degeneracy, and note the following corollary of theorem \ref{thm:main}.

\begin{corollary}
\label{thm:klein}
 An $(N,k)$ Adinkra with associated code $C$ has Klein flip degeneracy if and only if the all-1 codeword $(11\ldots1) \in C$. This in turn occurs only if $N\equiv 0 \pmod{4}$.
\end{corollary}

\begin{corollary}
 The automorphism group of a valise $(N,k)$ Adinkra $G(V,E)$ has size $2^{N-2k-a}$, with $2^{k+a}$ orbits of equal size, where $a = 0$ if $C$ contains the all-1 codeword, and $a = 1$ otherwise.
\end{corollary}

\begin{proof}
 $G$ has $2^{N-k}$ nodes. By theorem \ref{thm:main}, disregarding the vertex bipartition there are $2^k$ orbits of equal size, partitioning the vertex set. Including the bipartition, these orbits are split in 2 once more whenever they contain nodes of variable weight modulo 2 (i.e. whenever $C$ does not contain the all-1 codeword). 
\end{proof}

\setcounter{equation}{0}
%%%%%%%%%%%%%%%%%%%%%%
\section{Characterizing Adinkra Degeneracy}
\label{sec:gamma_matrix}

In the work of \cite{Gates09} the Klein flip degeneracy in the $N=4$ case of Example \ref{ex:N=4} was investigated. It was shown that Adinkras belonging to these two equivalence classes can be distinguished via the trace of a particular matrix derived from each Adinkra. We will briefly introduce these results, and generalize them to general $(N,k)$ Adinkras.

Throughout the following section we will consider an $(N,k)$ Adinkra $G(V,E)$ with corresponding code $C$. Note that $G$ has $2^{N-k}$ vertices, each of degree $N$.

\subsection{Notation}

\begin{definition}
 The \emph{adjacency matrix} $A(G)$ (or simply $A$ where the relevant Adinkra is clear from the context) is a $2^{N-k} \times 2^{N-k}$ symmetric matrix containing all the information of the graphical representation, except the vertex coloring associated with height assignments. Each element of the main diagonal represents either a boson or a fermion, denoted by $+i$ for bosons and $-i$ for fermions. The off-diagonal elements represent edges of $G$, numbered according to edge dimension, with sign denoting dashedness (e.g. if position $A_{x,y} = -4$, there is a dashed edge of the fourth edge color between boson/fermion $x$ and fermion/boson $y$). Hence the sign of the main diagonal represents the vertex bipartition, the sign of off-diagonal elements represents the edge parity, the absolute value of off-diagonal elements represents edge color, and the position of the elements encodes the topology of the Adinkra.
\end{definition}

\begin{example}
\label{ex:adjmat}
 Consider the 3-cube Adinkra of Example \ref{ex:ad2}.

 If we order the vertex set into bosons and fermions, this Adinkra can be represented by the matrix:
\begin{align}
 \textrm{A} = \begin{bmatrix}
 i&0&0&0&1&2&3&0\\
 0&i&0&0&-2&1&0&3\\
 0&0&i&0&-3&0&1&-2\\
 0&0&0&i&0&-3&2&1\\
 1&-2&-3&0&-i&0&0&0\\
 2&1&0&-3&0&-i&0&0\\
 3&0&1&2&0&0&-i&0\\
 0&3&-2&1&0&0&0&-i
 \end{bmatrix}
\end{align}
 Note that the property of symmetry possessed by the adjacency matrix arises naturally out of the definition, and requiring that this property be upheld imposes no further restrictions in and of itself.
\end{example}

We define L and R matrices similarly, as in \cite{Gates09}, to represent a single edge dimension of the Adinkra. L and R matrices encode this edge dimension, with rows corresponding to bosons and fermions respectively. As each edge dimension is assigned a separate matrix, the elements corresponding to edges are all set to $\pm1$. Then the Adinkra of example \ref{ex:adjmat} has L matrices:

\be
\textrm{L}_1 = \left[\begin{array}{cccc}
 1 & 0 & 0 & 0 \\
 0 & 1 & 0 & 0 \\
 0 & 0 & 1 & 0 \\
 0 & 0 & 0 & 1
\end{array}\right] , \quad
\textrm{L}_2 = \left[\begin{array}{cccc}
 0 & 1 & 0 & 0 \\
 -1 & 0 & 0 & 0 \\
 0 & 0 & 0 & -1 \\
 0 & 0 & 1 & 0
\end{array}\right] , \quad
\textrm{L}_3 = \left[\begin{array}{cccc}
 0 & 0 & 1 & 0 \\
 0 & 0 & 0 & 1 \\
 -1 & 0 & 0 & 0 \\
 0 & -1 & 0 & 0
\end{array}\right].
\ee

These matrices can be read directly off the top right quadrant of the adjacency matrix. Similarly, the R matrices can be read off the lower left quadrant. In this case, we have $\textrm{R}_\textrm{i} = \textrm{L}_\textrm{i}^\textrm{T}$, where T denotes the matrix transpose. Note that L and R matrices taken directly from the adjacency matrix will always be related via matrix transpose. Furthermore, since we are considering only off-shell supermultiplets, these matrices must also be square. We will assume that all subsequent L and R matrices are related in this way. We define a further object, $\gamma_\textrm{i}$, to be a composition of L and R matrices of the form 
\begin{align}
\label{eqn:gamma}
 \gamma_\textrm{i} = \left[\begin{array}{cc}
 0 & \textrm{L}_\textrm{i} \\
 \textrm{R}_\textrm{i} & 0                    
 \end{array} \right]
\end{align}
Further details regarding these objects can be found in \cite{Gates09}. In particular, it was shown in \cite{Gates09} that the Klein flip degeneracy of several $(4,1)$ Adinkras can be characterized by the trace of quartic products of these matrices, according to the formula
\begin{align}
\label{eqn:trace}
 &\textrm{Tr}(\textrm{L}_i (\textrm{L}_j)^\textrm{T} \textrm{L}_k (\textrm{L}_l)^\textrm{T}) = 4 \left( \delta_{i j} \delta_{k l} - \delta_{i k} \delta_{j l} + \delta_{i l} \delta_{j k} +  \chi{}_{{}_0} \epsilon_{i j k l} \right),
\end{align}
where $\delta$ represents the Kroenecker delta and $\epsilon$ the Levi-Civita symbol, and where $
\chi{}_{{}_0}  = \pm 1$ distinguishes between $(4,1)$ Adinkras belonging to the two equivalence classes of Example \ref{ex:N=4}. It is also conjectured that this same method for distinguishing equivalence classes can be generalized to all values of the parameters $N$ and $k$.

\subsection{New Results}

In order to generalize equation \ref{eqn:trace} to higher $N$, we wish to determine the form taken by the trace of products of these $\gamma$ matrices, and establish exactly when this can be used to partition Adinkras into their equivalence classes. Firstly we note the following properties of Adinkras:

\begin{lemma}
 The cycles of an $N$-cube Adinkra consist entirely of paths containing each edge dimension $0$ times (modulo 2).
\end{lemma}

\begin{lemma}
\label{thm:cycle}
 The cycles of an $(N,k)$ Adinkra with associated code $C$ consist of paths in which:
 \begin{itemize}
  \item[(i)] Each edge dimension is traversed $0 \pmod{2}$ times.
  \item[(ii)] At least one edge dimension is traversed $1 \pmod{2}$ times.
 \end{itemize}
 In case (ii), the set of such edge dimensions traversed $1 \pmod{2}$ times corresponds to a codeword in $C$.
\end{lemma}
In fact, the codeword $(00\ldots0)$ is in any such code $C$, so all cycles have this property, however case (i) is considered a trivial instance.

Consider the product of $t$ $\gamma$ matrices of $G$, denoted by $M = \gamma_{i_1} \gamma_{i_2} ... \gamma_{i_t}$. As we are considering only L and R matrices such that $\textrm{L}_\textrm{i} = \textrm{R}_\textrm{i}^\textrm{T}$, all $\gamma$ matrices defined as in equation \ref{eqn:gamma} will be real and symmetric. This leads to the following property for $M$.

\begin{lemma}
\label{thm:cycle2}
 The elements of the main diagonal of $M$ are non-zero if and only if the path $P_M = (i_1, i_2, \ldots, i_t)$ represents a closed loop within the Adinkra. This can occur one of two ways:
 \begin{itemize}
  \item[(i)] Trivially, if each edge dimension contained in $p$ is present $0 \pmod{2}$ times,
  \item[(ii)] If $p$ corresponds to a codeword (or set of codewords) in $C$.
 \end{itemize}
\end{lemma}

\begin{proof}
 Consider the case $P_M = (i,j)$. If $i=j$, $M$ is trivially the identity, hence (i) holds. If $i \ne j$, then $M_{x,y} \ne 0$ if and only if the vertices $x$ and $y$ are connected via a path $(i,j)$ (i.e. if $x = i.j.y$). Extending to general $M$, if $M_{x,x} \ne 0$ for some $x \in V$, this implies that $x$ is connected to itself via the path $P_M$, a closed loop / cycle in $G$. Hence lemma \ref{thm:cycle} completes the proof.
\end{proof}

\begin{corollary}
 $M_{x,x} \ne 0$ for some $x \in V$ if and only if this is true for all $x \in V$.
\end{corollary}

\begin{lemma}
\label{thm:cycle3}
 In case (ii) of lemma \ref{thm:cycle2}, with $M_{x,x} = \pm 1$ for all $x \in V$, $M_{x,x} = M_{y,y}$ if and only if $x$ and $y$ have the same inner product with the codeword corresponding to $P_M$.
\end{lemma}

\begin{proof}
 Consider the product of Clifford generators corresponding to the path $P_M$. Since (ii) holds, after cancelling repeated elements we are left with some codeword $p \in C$, such that $p = a P_M$, where $a = \pm 1$, depending on whether $P_M$ corresponds to an even or odd permutation of $p$, relative to shifting and cancelling of Clifford generators. Since $C$ is a group, either
 \begin{itemize}
  \item[(i)] $p \in D$, a standard form generating set of $C$, or
  \item[(ii)] $p = (g_1 \boxplus g_2 \boxplus \ldots \boxplus g_r)$, where $g_i \in D$, for any such $D$.
 \end{itemize}
 Assume (i) holds. Then $p = (p1,p2,\ldots,p_t)$ such that $p_i \le (N-k)$ for $i \ne t$, and $p_t > (N-k)$, with the first $(N-k)$ edge dimensions defined relative to the particular standard form generating set $D$ being considered. In other words, $p_t = p_1 \cdot p_2 \cdot \ldots \cdot p_{t-1}$. Then $\forall x \in G, x = (x_1 x_2 \ldots x_{N-k})$, the sign of $M_{x,x}$, omitting a factor of $\frac{a-1}{2}$, is given by (substituting the formulas for $\pi_i(x)$ of Section \ref{sec:notation2})
 \begin{align}
  (p_1 \cdot x &+ p_2 \cdot x + \ldots + p_{t-1} \cdot x) + x \cdot p_t & \nonumber \\
  &\equiv p_2(x_1) + p_3(x_1+x_2) + \ldots + p_{t-1}(x_1+x_2+\ldots+x_{t-2})  \nonumber \\
  & {~~~~} %\hspace{5.5cm} 
  + p_1(x_2+\ldots+x_{t-1}) + \ldots + p_{t-2}(x_{t-1}) + |x|  \nonumber \\
  &\equiv (|p|-1)|x| - \sum_{i=1}^{t-1}(p_i x_i) + |x|  \nonumber \\
  &\equiv |x||p| + \langle x,p\rangle \pmod{2}
 \end{align}
 Note that $|p| \equiv 0 \pmod{4}$, so the first term disappears. The same arguments can be followed to show that this holds for case (ii) also. Then since the factor of $\frac{a-1}{2}$ is constant for all $x \in G$, we have for all $x,y \in G$, $M_{x,x} = M_{y,y}$ if and only if $\langle x,p\rangle = \langle y,p\rangle$.
\end{proof}

One immediate corollary of the preceeding lemma is that the trace of $M$ will vanish whenever $P_M$ corresponds to a codeword in $C$. In fact, Tr($M$) will only be non-zero in the trivial case where the path $P_M$ consists of a set of pairs of edges of the same color. These are the paths corresponding to the Kroencker delta terms of equation \ref{eqn:trace}. In particular, this implies that Tr($M$) cannot distinguish between equivalence classes of Adinkras directly. However if we instead consider powers of L and R matrices, the preceeding lemma suggests a direct generalization of equation \ref{eqn:trace} to all $(N,k)$ Adinkras.

Given an $(N,k)$ Adinkra, consider the product of $N$ L matrices, $\textrm{L}_{i_1} \textrm{L}_{i_2}^\textrm{T} \ldots \textrm{L}_{i_t}^\textrm{T}$. Denoting the path $p = (i_1,i_2,\ldots,i_N)$, we define the value $\sigma_p$ such that $\sigma_p = 0$ if there exists an edge dimension in $p$ that is present $1 \pmod{2}$ times, and $\sigma_p = a$ otherwise. Here $a = \pm 1$, corresponding to the \emph{sign} of the related product of Clifford generators (since $p$ consists of pairs of edges of the same color, this product of Clifford generators equals $\pm 1$). Then we have the following result.

\begin{lemma}
\label{thm:chinull}
 The trace of $\textrm{L}_{i_1} \textrm{L}_{i_2}^\textrm{T} \ldots \textrm{L}_{i_{t-1}} \textrm{L}_{i_t}^\textrm{T}$, where $p = (i_1,i_2,\ldots,i_N)$, equals
 \begin{align}
  2^{N-k-1} \,  (\sigma_p +  \chi{}_{{}_0}  \epsilon_p),
 \end{align}
 where $ \chi{}_{{}_0} = \pm 1$ depending on the equivalence class of the $(N,k)$ valise Adinkras.
\end{lemma}

\begin{proof}
 The $\sigma_p$ term follows directly from lemmas \ref{thm:cycle2} and \ref{thm:cycle3}. In Section \ref{sec:main}, we show that $(N,k)$ Adinkras with the same associated code $C$ are all in a single equivalence class, except where $(11\ldots1) \in C$. In this case, the two equivalence classes are related via the Klein flip operation, exchanging bosons and fermions. Moreover, since $c=(11\ldots1) \in C$, the Klein flip operation switches the sign of $\langle x,c\rangle$, for each $x \in V$. Then by lemma \ref{thm:cycle3}, the Adinkras from different equivalence classes correspond to $ \chi{}_{{}_0}$ values of opposite sign.
\end{proof}

Note that if $N=4$ and $p=(i,j,k,l)$, then $\sigma_p = \delta_{i j} \delta_{k l} - \delta_{i k} \delta_{j l} + \delta_{i l} \delta_{j k}$, and the result of equation \ref{eqn:trace} follows. These results can be simplified in some respects if we consider the products of $\gamma$ matrices instead. Recall the definition $M = \gamma_{i_1} \gamma_{i_2} ... \gamma_{i_t}$, where $P_M = (i_1,i_2,\ldots,i_t)$, and the fermion number operator $(-1)^F$. Instead of taking the trace of $M$, consider the trace of $M . (-1)^\textrm{F}$. We have
\be \gamma_\textrm{i} = \left[\begin{array}{cc}
 0 & \textrm{L}_\textrm{i} \\
 \textrm{R}_\textrm{i} & 0                    
 \end{array} \right] \quad \textrm{and} \quad
(-1)^\textrm{F} = \left[\begin{array}{cc}
 \textrm{I} & 0 \\
 0 & -\textrm{I}
\end{array}\right],\ee
and we are considering only $\textrm{L} = \textrm{R}^\textrm{T}$, so for even $t$,
\begin{align}
 M = \left[\begin{array}{cc}
 \textrm{L}_{i_1} \textrm{L}_{i_2}^\textrm{T} \ldots \textrm{L}_{i_t}^\textrm{T} & 0 \\
 0 & \textrm{L}_{i_1}^\textrm{T} \textrm{L}_{i_2} \ldots \textrm{L}_{i_t}
 \end{array} \right].
\end{align}
Replacing L matrices by R matrices in lemma \ref{thm:chinull} simply changes the sign of
$ \chi{}_{{}_0}$, so for an $(N,k)$ Adinkra,
\begin{align}
 \textrm{Tr}(M \cdot (-1)^\textrm{F}) = 2^{N-k} \,  \chi{}_{{}_0} \epsilon_p.
\end{align}

\setcounter{equation}{0}
%%%%%%%%%%%%%%%%%%%%%%%%%%%
\section{Identifying Isomorphism Classes of Adinkras}
\label{sec:iso}

The results up to this point deal with equivalence classes of valise Adinkras valise case, where we consider only 2-level Adinkras. For the non-valise case, we require a method of partitioning general $(N,k)$ Adinkras into their isomorphism classes. The essential problem in establishing such a method is in classifying the topology of an Adinkra relative to the automorphism group of its underlying 1-level Adinkra (in which vertex labels are ignored). We might consider simply partitioning the vertex set into the orbits of this underlying automorphism group, and then classifying each height by the number of vertices of each orbit that it possesses. However this method will clearly be insufficient, as it ignores the relative connectivity between vertices at different heights. For example, the two Adinkras of Figure \ref{fig:nonvalise} below are in the same equivalence class by use of this set of defintions. They also have the same number of vertices from each orbit at each height, and yet they are clearly non-isomorphic; no relabeling of the vertices or vertex switching operations can map between them.

\begin{figure}[h]
 \begin{center}
 \includegraphics[width=4.5cm]{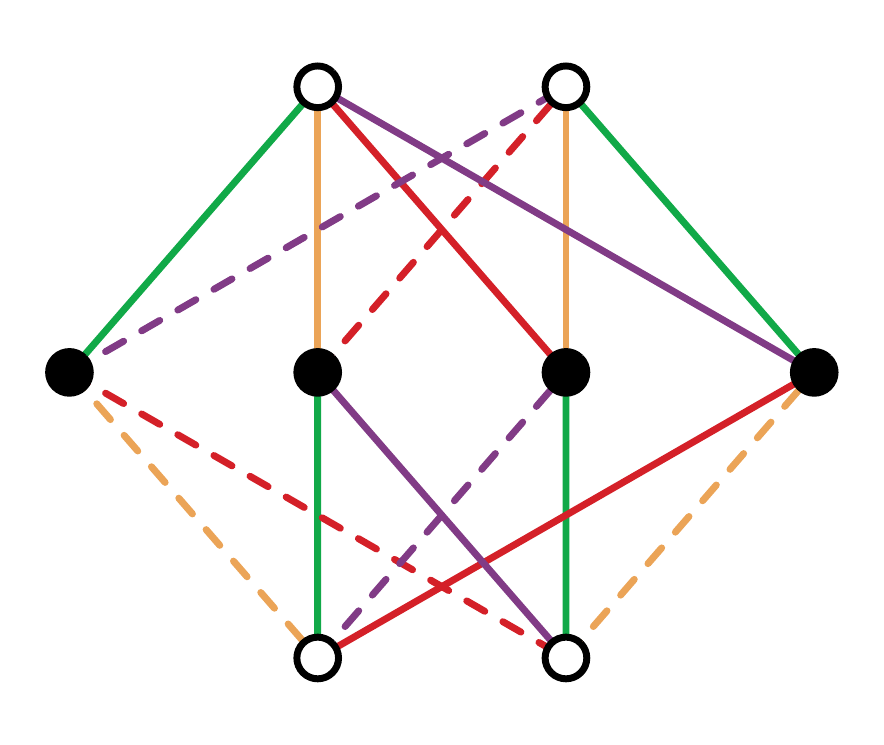} \hspace{1cm} \includegraphics[width=4.5cm]{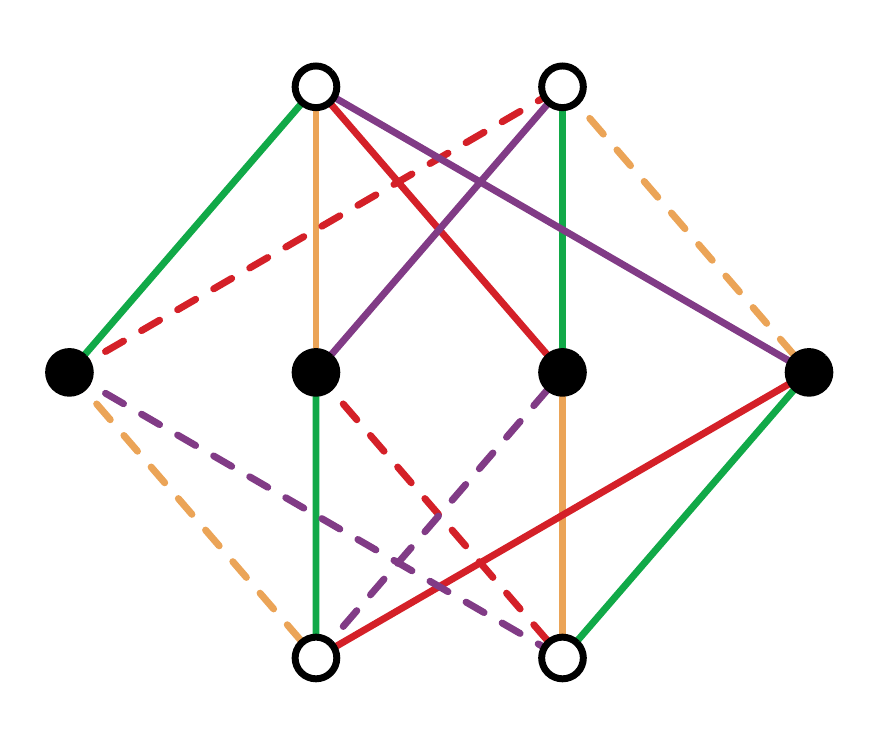}
 \caption{\label{fig:nonvalise} Two equivalent, non-isomorphic Adinkras with the same values of $\mu(h,a)$, for all possible values of $h$ and $a$.}
 \end{center}
\end{figure}

However even this simple method does come close to partitioning non-valise Adinkras into their isomorphism classes. Recall that a pointwise stabilizer of the automorphism group of any Adinkra is the identity - no non-trivial automorphisms exist that fix any vertex\symbolfootnote[2]{Here the term \emph{non-trivial} refers to the corresponding permutation of the vertex set. An automorphism 
\newline $~~~\,~~$
which switches vertices but leaves the ordering of the vertex set unchanged is considered trivial.}. In fact fixing any vertex of an Adinkra yields a natural, canonical ordering of the vertex set, relative to some ordering of the edge dimensions, according to the following construction.
\begin{construction}
 Suppose we are given an $(N,k)$ Adinkra $G(V,E)$, together with an ordering of the edge dimensions from $i_1,i_2,\ldots,i_N$. Then fix (choose) any vertex $v \in V$. We define an ordering $\lambda_v$ of the vertex set relative to $v$, where $\lambda_v : V \mapsto \left[ 2^{N-k} \right]$, in the following way.
 \begin{itemize} 
  \item $\lambda_v(v) = 1$.
  \item Order the neighbours of $v$ from 2 to $N+1$ according to the ordering of the corresponing edge dimensions ($\lambda_v(i_n v) = 1+n$, where $1 \le n \le N$).
  \item Repeat this for vertices at distance 2, beginning with the neighbours of $i_1 v$, and ending with the neighbours of $i_N v$.
  \item Repeat this similarly for vertices at each distance, until all vertices have been assigned an ordering in $\left[ 2^{N-k} \right]$.
 \end{itemize}
\end{construction}
In other words, $\lambda_v$ firstly orders the vertex set according to distance from $v$, then for each of these sets, each vertex $x$ is assigned an ordering based on the lexicographically smallest path from $v$ to $x$.

To formalize the partitioning of heights discussed above, consider an $(N,k)$ Adinkra $G(V,E)$ with $h$ different height assignments. Let $\Gamma_G$ be the automorphism group of the corresponding 1-level Adinkra (ignoring the vertex coloring of $G$). Relative to a particular ordering of the edge dimensions, and a particular generating set of the associated doubly even code, order the $2^k$ orbits of $\Gamma_G$ according to $\tau : V \mapsto \left[2^k\right]$. We then define $\mu_G(h,a)$ to be the number of vertices at height $h$ belonging to the $a^\textrm{th}$ orbit of $\Gamma_G$, such that
\be
\mu_G(h,a) = |\{ v \in V : \textrm{hgt}(v) = h, \tau(v) = a \}|.
\ee
Consider a pair of Adinkras $G(V,E)$ and $H(V,E')$ belonging to the same equivalence class, both having $t$ distinct heights. Then $G$ and $H$ have the same associated code $C$. If $\mu_G(h,a) = \mu_H(h,a)$, for all $1 \le h \le t$, relative to a given edge-color ordering and generating set of $C$, then choose any vertex $v$ from each Adinkra belonging to the same orbit and height. Relative to this vertex $v$, consider the unordered set
\be
\textrm{cert}_G(v) = \{(\lambda_v(x),\textrm{hgt}(x),\tau(x) ) : x \in V \}.
\ee

\begin{theorem}
 \label{thm:iso}
 Given two Adinkras $G$ and $H$ and a vertex $v$ from each with the properties described above. Then $\textrm{cert}_G(v) = \textrm{cert}_H(v)$ if and only if $G$ and $H$ are isomorphic.
\end{theorem}

\begin{proof}
 If $G$ and $H$ are isomorphic then this is trivially true. Conversely, assume $\textrm{cert}_G(v) = \textrm{cert}_H(v)$. Then since $G$ and $H$ belong to the same equivalence class, and $\tau_G(v) = \tau_H(v)$, there exists an isomorphism $\gamma$ mapping $G$ to $H$ (ignoring the vertex colorings). Then $\gamma$ extends to a full isomorphism (including the vertex colorings) if it preserved height assignments. This follows directly from $\textrm{cert}_G(v) = \textrm{cert}_H(v)$, hence $G$ and $H$ are isomorphic.
\end{proof}

The preceeding theorem provides an efficient method of classifying Adinkras according to their isomorphism class. 
Note that the classification is relative to a given ordering of the edge colors, and requires a knowledge of the associated doubly even code. In cases where the associated code is unknown, lemma \ref{thm:cycle} suggests an efficient method for finding the code, and hence relating a given $(N,k)$ Adinkra to its `parent' $N$-cube Adinkra, in the sense of theorem \ref{thm:quotient}. In particular, lemma \ref{thm:cycle} implies the following result.

\begin{corollary}
 Given an $(N,k)$ Adinkra with related code $C$, the codewords of $C$ correspond exactly to the cycles of $G$ in which at least 1 edge dimension appears once (modulo 2).
\end{corollary}

We provide an example of the above certificates below. Consider the two Adinkras of Figure \ref{fig:cert1}. These are both height-3, $(5,1)$ Adinkras. They are in the same equivalence class, however by calculating the above certificate for each we will show that they are non-isomorphic.

\begin{figure}[h]
 \begin{center}
 \includegraphics[width=9.5cm]{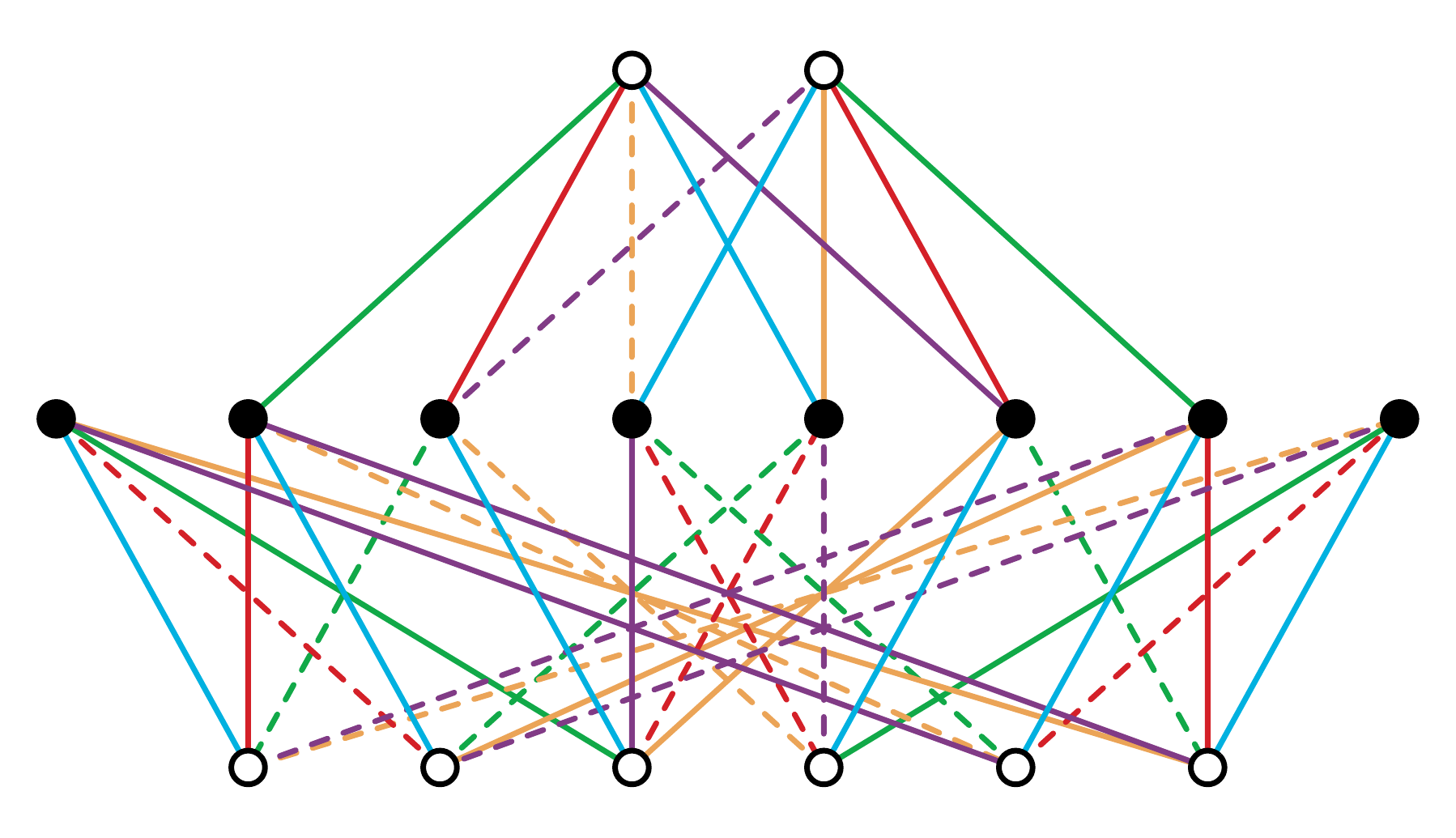} \hspace{1cm} \includegraphics[width=9.5cm]{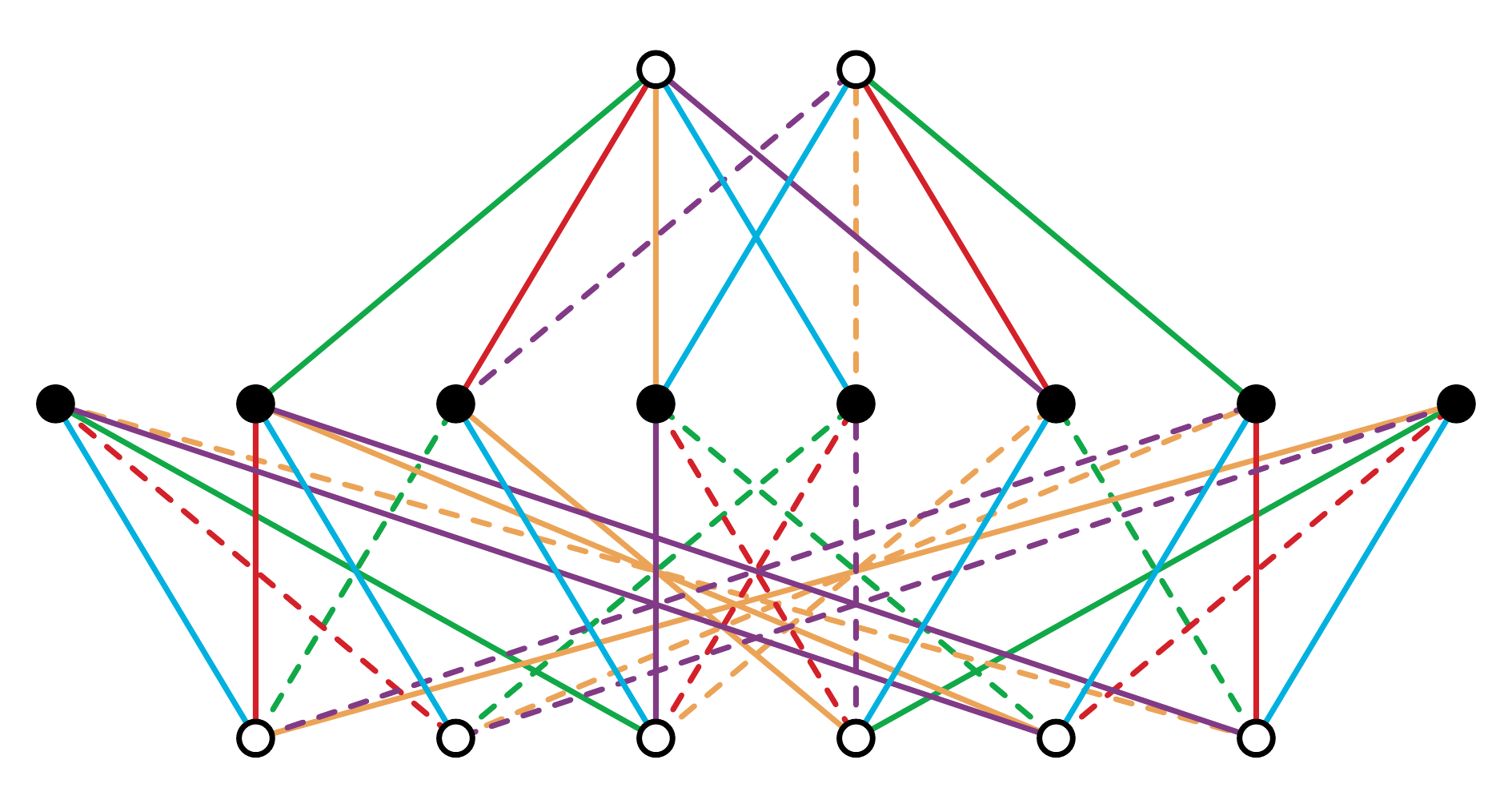}
 \caption{\label{fig:cert1} Two equivalent, non-isomorphic $(5,1)$ Adinkras. They differ precisely in the switching state of the yellow edges.}
 \end{center}
\end{figure}

To analyse the above Adinkras, we first order the edge colors from green to purple, such that in the bit-string representation, green corresponds to the first bit, and purple to the fifth. Alternatively, the green edges are associated with the first Clifford generator, and the purple edges to the fifth generator. As these are $(5,1)$ Adinkras, they have an associated $(5,1)$ double even code. By inspection, or by explicitly finding the non-trivial cycles in the Adinkras, we note that the associated code is $(11110)$\symbolfootnote[2]{Note that `first' bit of the bit-string is the right-most bit.}, corresponding to a 4-cycle comprising edge colors $\{\textrm{red, yellow, blue, purple}\}$. By theorem \ref{thm:main}, the induced 1-level Adinkra has two orbits, corresponding to the sets of nodes with the same parity (inner product modulo 2) relative to this code. Denote the two Adinkras by $G$ and $H$ respectively. Then by applying the results of Section \ref{sec:main}, we obtain $\mu$ values of:

\begin{align}
\label{ex:orbits}
 &\mu_G(3,1) = 2 & & \mu_H(3,1) = 0 \notag\\
 &\mu_G(3,2) = 0 & & \mu_H(3,2) = 2 \notag\\
 &\mu_G(2,1) = \mu_G(2,2) = 4 & & \mu_H(2,1) = \mu_H(2,2) = 4 \\
 &\mu_G(1,1) = 2 & & \mu_H(1,1) = 4 \notag\\
 &\mu_G(1,2) = 4 & & \mu_H(1,2) = 2 \notag
\end{align}

In other words, the two height-3 vertices of $G$ are in a different orbit to those of $H$. Hence $\textrm{cert}_G \ne \textrm{cert}_H$, and the two Adinkras are in different isomorphism classes. Additional examples illustrating this method of partitioning vertices into isomorphism classes are presented in Appendix \ref{app:iso}.

\setcounter{equation}{0}
%%%%%%%%%%%%%%%%%%%%%%%%%
\section{Conclusion}

This work provides a graph theoretic characterization of Adinkras, in particular classifying their automorphism groups according to an efficiently computable set of local parameters.
In the current work, we have been able to make a number of comparisons to previous work.  The connection between Adinkras
and codes \cite{Doran083} has been reexamined and found to be robust.  However, we have also exploited this connection to utilize the {\em {standard form}} leading to more computational efficient
algorithms for the study of Adinkras.  As well, the observations based on matrix methods
used within the context of d $=$ ${ N}$ $=$ 4
\cite{Gates09} have now been extended by a formal proof to all values of d and $N$.  Also as
emphasized in appendix A, numerical studies up to values of $N$ = 16 provide additional concurrence.
These results support the proposal that $ \chi{}_{{}_0}$ `chi-null' is a class valued
function defined on valise Adinkras.  

All non-valise Adinkras, through a series of node raising
and lowering can be brought to the form of a valise Adinkra.  In this sense $ \chi{}_{{}_0}$ 
is defined for all Adinkras.  However, for non-valise Adinkra, $ \chi{}_{{}_0}$ is not sufficient
to define classes.  For this purpose, the new certificate $\m_{\cal A}(h, \, a)$, where $\cal A$ is
an arbitrary Adinkra, seems to fill in a missing gap.

It is the work of future investigations to explore whether these tools ($ \chi{}_{{}_0}$ and
$\m_{\cal A}(h, \, a)$) are sufficient to attack the problem of the complete classification of
one-dimensional off-shell supersymmetrical systems.  One obvious future avenue of study is to investigate the
role codes play in Gnomon Adinkras.  This as part of continuing to attack the general
problem presents continuing challenges.

 \vspace{.55in}
 \begin{center}
 \parbox{4.3in}{{\it ``Mathematics: The science of skillful operations with concepts $\,~~$
 $~$~ and rules invented for this purpose.''\\ $\,~~$
 }\,--  E.\ Wigner}
 \end{center}

 \vspace{.55in}
 \noindent
{\bf Acknowledgements}\\[.1in] \indent
This work was partially supported by the National Science Foundation grants 
PHY-0652983 and PHY-0354401. This research was also supported in part by the 
endowment of the John S.~Toll Professorship and the University of Maryland Center for 
String \& Particle Theory.   Adinkras 
were drawn with the aid of the {\em  Adinkramat\/}~\copyright\,2008 by G.~Landweber.
We also extend an additional  note of appreciation to him for a critical reading of
our work and for sharing unpublished results with us. 
S.J.G. very gratefully wishes to acknowledge the University of Western Australia and 
especially the Institute for Advanced Study for warm hospitality (especially I.\ McArthur 
and S. Kuzenko) and a most stimulating location and atmosphere that marked the initiation 
of this study. B.L.D and J.B.W. would also like to acknowledge support from The University
of Western Australia, in particular the Hackett Scholarship.

\setcounter{equation}{0}
%%%%%%%%%%%%%%%%%%%%%%%%%
\newpage
\appendix
\section{Numerical Results and Methods}
\label{app:numerical}

The automorphism group results of the preceeding sections were also verified numerically, independently to the analytical results. All $(N,k)$ Adinkras up for $N \le 16$ were produced, and the related automorphism group and equivalence classes were calculated for each such Adinkra. The results found were consistent with the analytical results described in the earlier sections. In particular:

\begin{itemize}
 \item All $(N,k)$ Adinkras with the same associated code $C$ were found to be in the same equivalence class, except in the cases where $(11\ldots1) \in C$. In these cases, the Adinrkas were split into two equivalence classes, related via the Klein flip operation, as described in corollary \ref{thm:klein}.
 \item All 1-level $(N,k)$ Adinkra had $2^k$ orbits, each consisting of sets of vertices with the same set of inner products with the codewords in $C$, according to theorem \ref{thm:main}.
\end{itemize}

Of more interest as supplementary material may be some details regarding the methods used to obtain these numerical results. All doubly even $(N,k)$ codes up to $N=28$ (and many larger parameter sets) can be found online (Miller, \cite{rlmiller}). For each parameter set $(N,k)$, the two possibly inequivalent standard form Adinkras corresponding to the construction method of theorem \ref{thm:construction} were produced relative to each doubly even code on these parameters. As we were considering only equivalence classes and automorphism groups of \emph{valise} Adinkras, it sufficed to use the adjacency matrix form defined in Section \ref{sec:gamma_matrix} for this analysis, as height assignments need not be encoded in the representation. The orbits were then calculated from the adjacency matrices by forming a canonical form relative to each node via a set of switching operations and a permutation of the vertex set. A canonical form is defined to be a mapping $\pi$, consisting of a permutation of the vertex labels and set of vertex switching operations, such that for any two Adinkras $G$ and $H$, $\pi(G) = \pi(H)$ if and only if $G$ and $H$ are isomorphic. 

Then given a vertex $v$ and adjacency matrix $A$ of an $(N,k)$ Adinkra (together with an ordering of the edge colors), we define the following canonical form of $A$ relative to $v$.

\begin{itemize}
 \item[(i)] Permute the ordering of the vertices relative to their connections to $v$ and the ordering of edge colors, such that the vertices are ordered: 

$(v,i_1 v,\ldots, i_N v, i_2 i_1 v, \ldots, i_N i_1 v, i_3 i_2 v, \ldots, i_N i_{N-1} \ldots i_1 v)$.
 \item[(ii)] Switch vertices $(i_1 v, \ldots, i_N v)$ such that all edges of $v$ have even parity.
 \item[(iii)] Repeat for neighbors of vertex $i_1 v$.
 \item[(iv)] Repeat for each vertex, in the ordering above, such that edges appearing lexicographically earlier in the adjacency matrix are of even parity where possible.
\end{itemize}

This defines a unique switching state of the Adinkra, up to isomorphism. Furthermore it is a canonical form, in that any two vertices belonging to the same orbit result in identical matrix forms. To illustrate the above process, we consider an example $(6,2)$ Adinkra with associated code generated by
$\begin{pmatrix}
 1\,0\,1\,1\,1\,0\\[-1mm]
 0\,1\,1\,1\,0\,1\\[-1mm]
\end{pmatrix}$.
The standard form valise Adinkra is shown below, where the top leftmost vertex has bit-string $(0000)$, and its edges are ordered lexicographically from left to right.

\begin{center}
 \includegraphics[width=5.5cm]{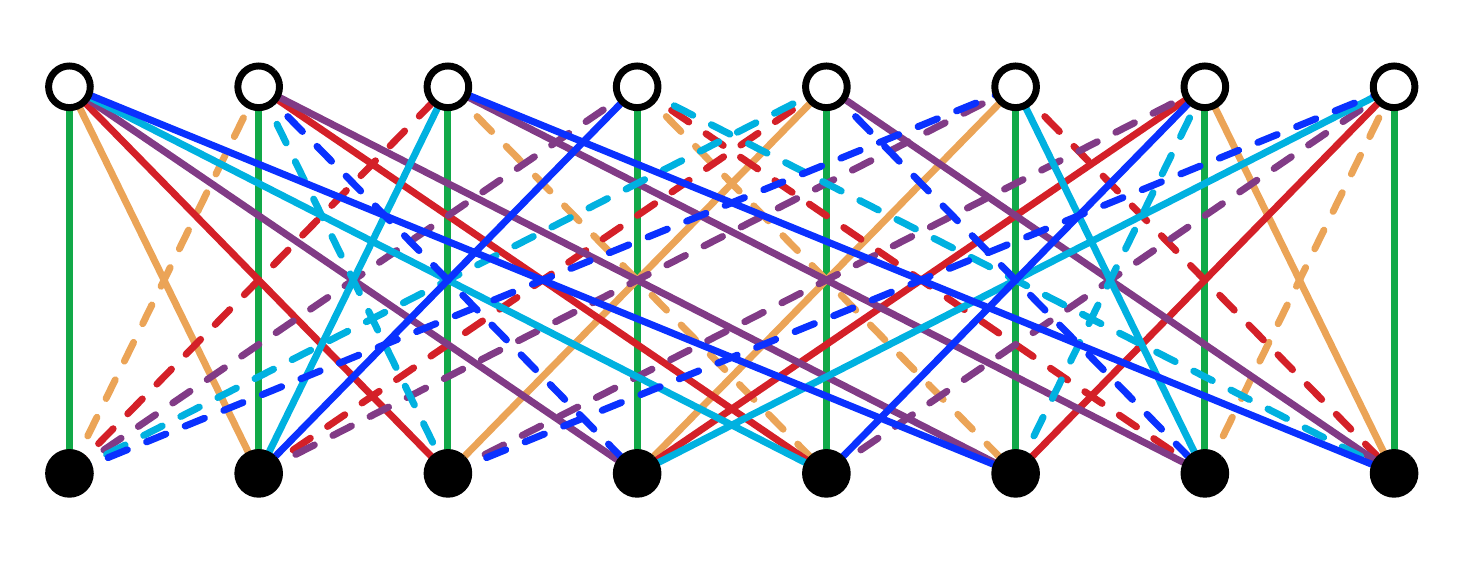}
\end{center}

Consider the process of converting this Adinkra to canonical form, relative to vertex $(0011)$ (the fifth white vertex from the left in the above figure). If we order the vertices by height first, then left to right, step (i) corresponds to the set of vertex lowering operations leading to the Adinkra:

\begin{center}
 \includegraphics[width=5.5cm]{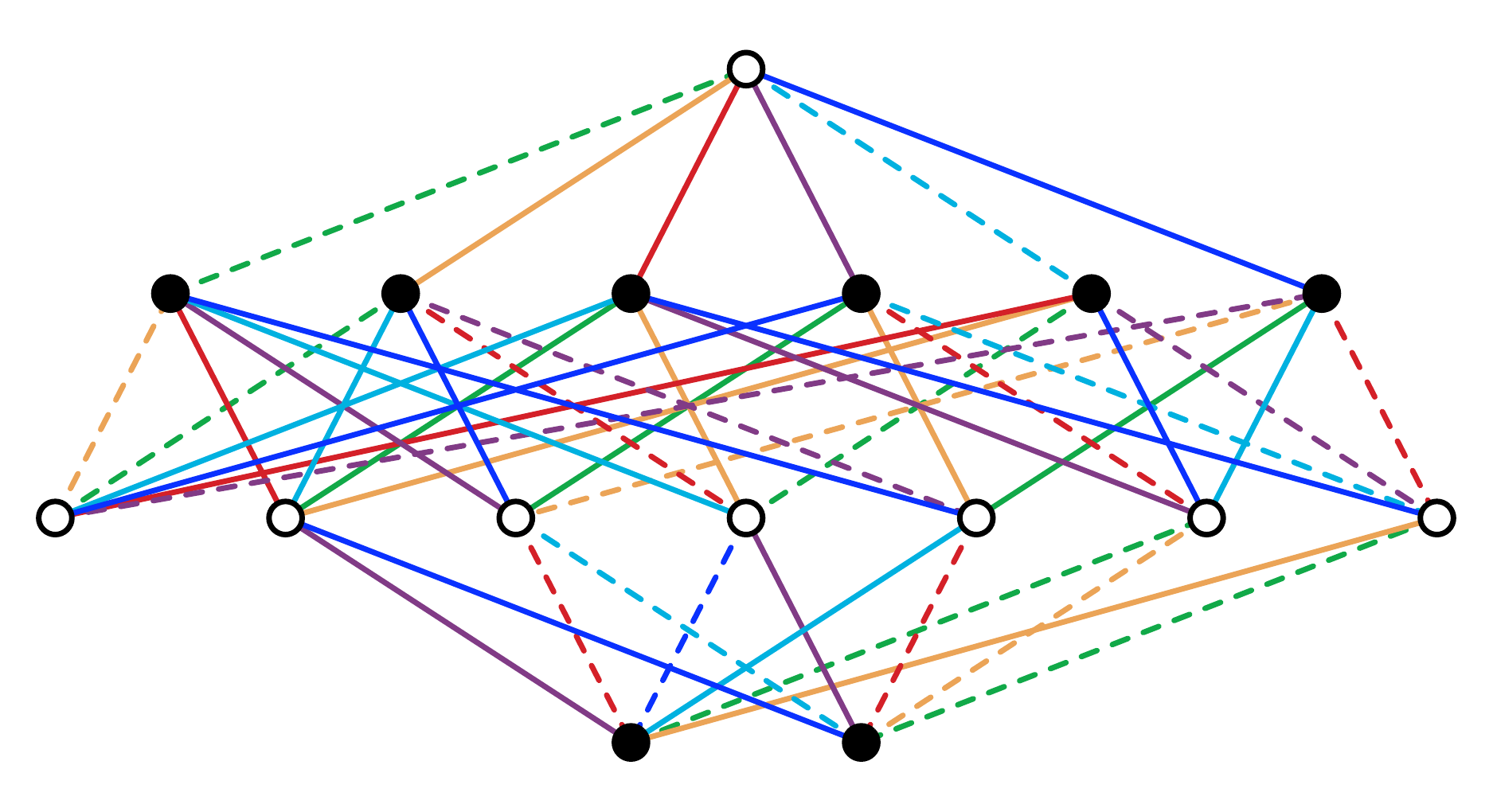}
\end{center}

Steps (ii) and (iii) correspond to a set of switching operations, permuting the switching state in the following stages:

 \begin{center}
 \begin{tabular}{ccc}
  \includegraphics[width=5.5cm]{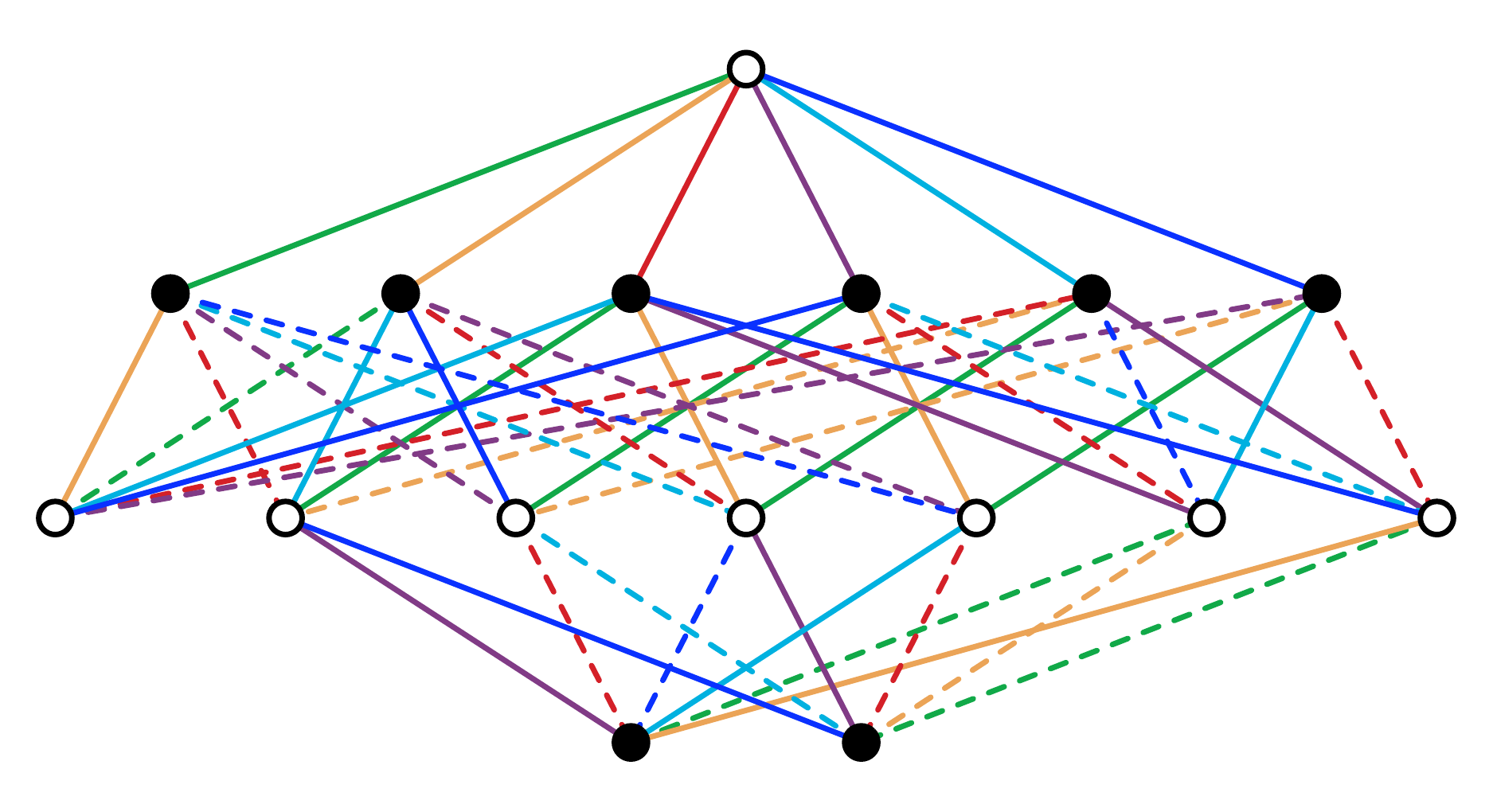} & \hspace{0.5cm} \raisebox{1.0cm}{\parbox{1.0cm}{$\Longrightarrow$}} \hspace{0.5cm} & \includegraphics[width=5.5cm]{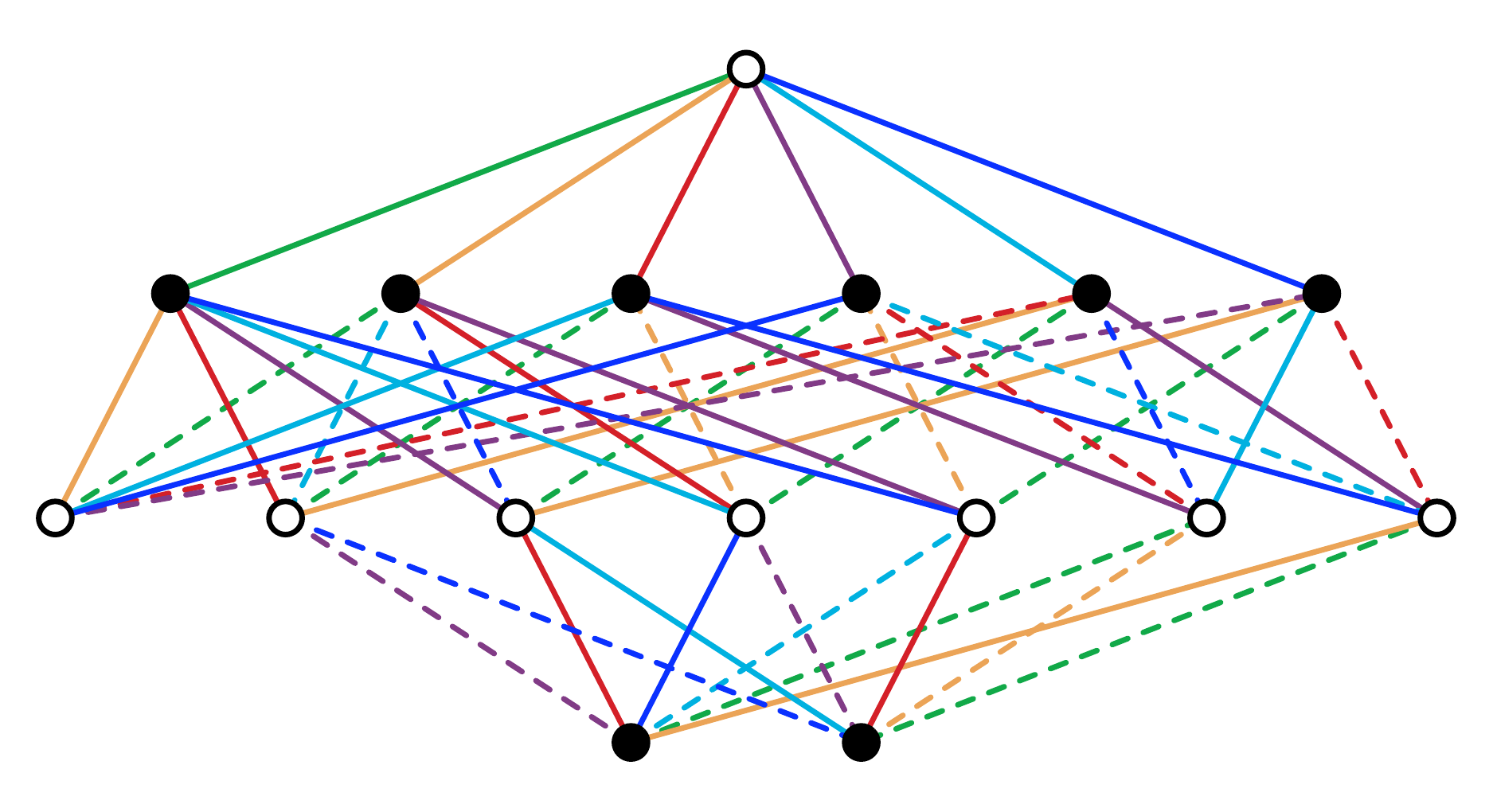}
 \end{tabular}
 \end{center}

Step (iv) then involves the remaining set of transformations:

 \begin{center}
 \begin{tabular}{ccc}
  \includegraphics[width=5.5cm]{source_12_2.pdf} & \hspace{0.5cm}  \raisebox{1.0cm}{\parbox{1.0cm}{$\Longrightarrow$}} \hspace{0.5cm} & \includegraphics[width=5.5cm]{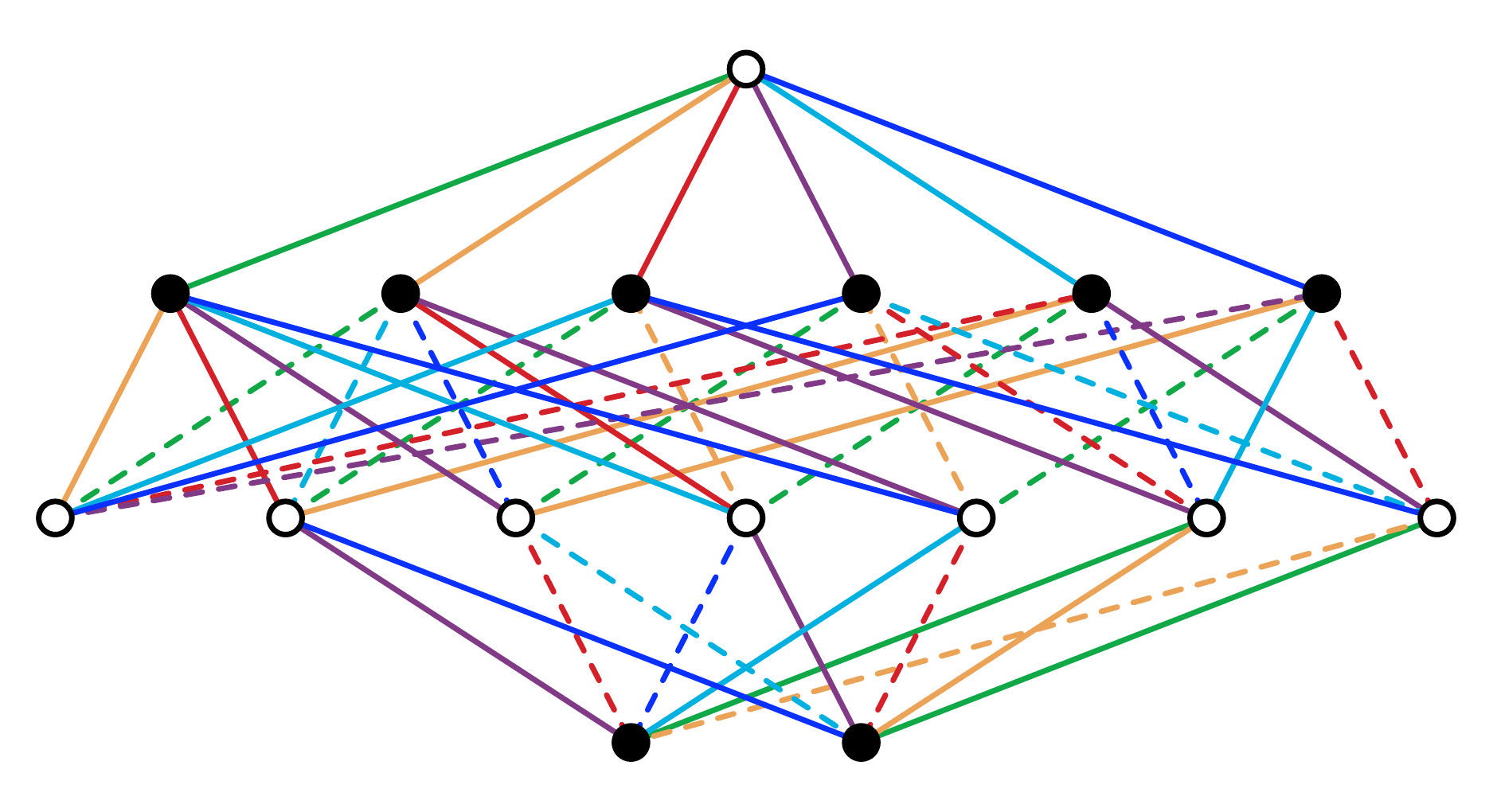}
 \end{tabular}
 \end{center}

At which point the Adinkra is now in canonical form. The particular switching state of this resulting Adinkra is unique to choices of source vertices in the same orbit, in this case the set of vertices ${(0000),(0011),(1110),(1101)}$.

%%%%%%%%%%%%%%%%%%%%%%%%%
\section{Partitioning into Isomorphism Classes: Further Examples}
\label{app:iso}

In \cite{counter} it was shown that simply recording the number of vertices at each height is not sufficient to characterize Adinkras. In particular, they provide examples of pairs of Adinkras which are in the same equivalence class, but not isomorphic, despite having the same number of vertices at each height. In this section, we will analyze the examples of \cite{counter} using the techniques of Section \ref{sec:iso}.

Firstly however, we note that a different definition of isomorphism is considered in \cite{counter}. Specifically, they consider the situation where permutations of the edge colors preserve isomorphism. Allowing this more general definition of isomorphism results in several changes to the results of the preceeding sections. In particular, note that the automorphism group of the $N$-cube Adinkra would simply become that of the $N$-cube: the hyperoctahedral group of order $2^{N}.N!$. A permutation of the edge colors corresponds to a reordering of the bit-string associated with each vertex (or equivalently a permutation of the Clifford generators associated each edge dimension). Hence the automorphism group of $(N,k)$ Adinkras will be extended according to the following lemma.

\begin{lemma}
\label{lem:app}
 Allowing permutations of the edge colors to preserve isomorphism, an $(N,k)$ Adinkra $G$ with associated doubly even code $C$ has an automorphism group (ignoring height assignments) of order $$\frac{2^N}{2^{2k+a}} |\textrm{Aut}(C)|,$$ where Aut($C$) is the automorphism group of the code $C$, and where $a = 1$ if $C$ contains the all-1 codeword, and $a = 0$ otherwise.
\end{lemma}

In other words, any permutations that correspond to symmetries of the code will extend naturally to automorphisms of the Adinkra. Conversely, if a permutation of the bit-string is not in the automorphism group of the code, then trivially it cannot be an automorphism of the Adinkra. Note that lemma \ref{lem:app} applies trivially to $N$-cube Adinkras, for which $k = 0$ and $C = (00\ldots0)$, hence $|\textrm{Aut}(C)| = N!$. 

In \cite{counter}, two pairs of equivalent but non-isomorphic Adinkras are presented, each pair having having the same number of vertices at each height. The first pair comprises two height 3, $(5,1)$ Adinkras, isomorphic (up to a permutation of edge colors) to those below. Note that the second Adinkra is identical to the first Adinkra of Figure \ref{fig:cert1}.

\begin{center}
 \includegraphics[width=9.5cm]{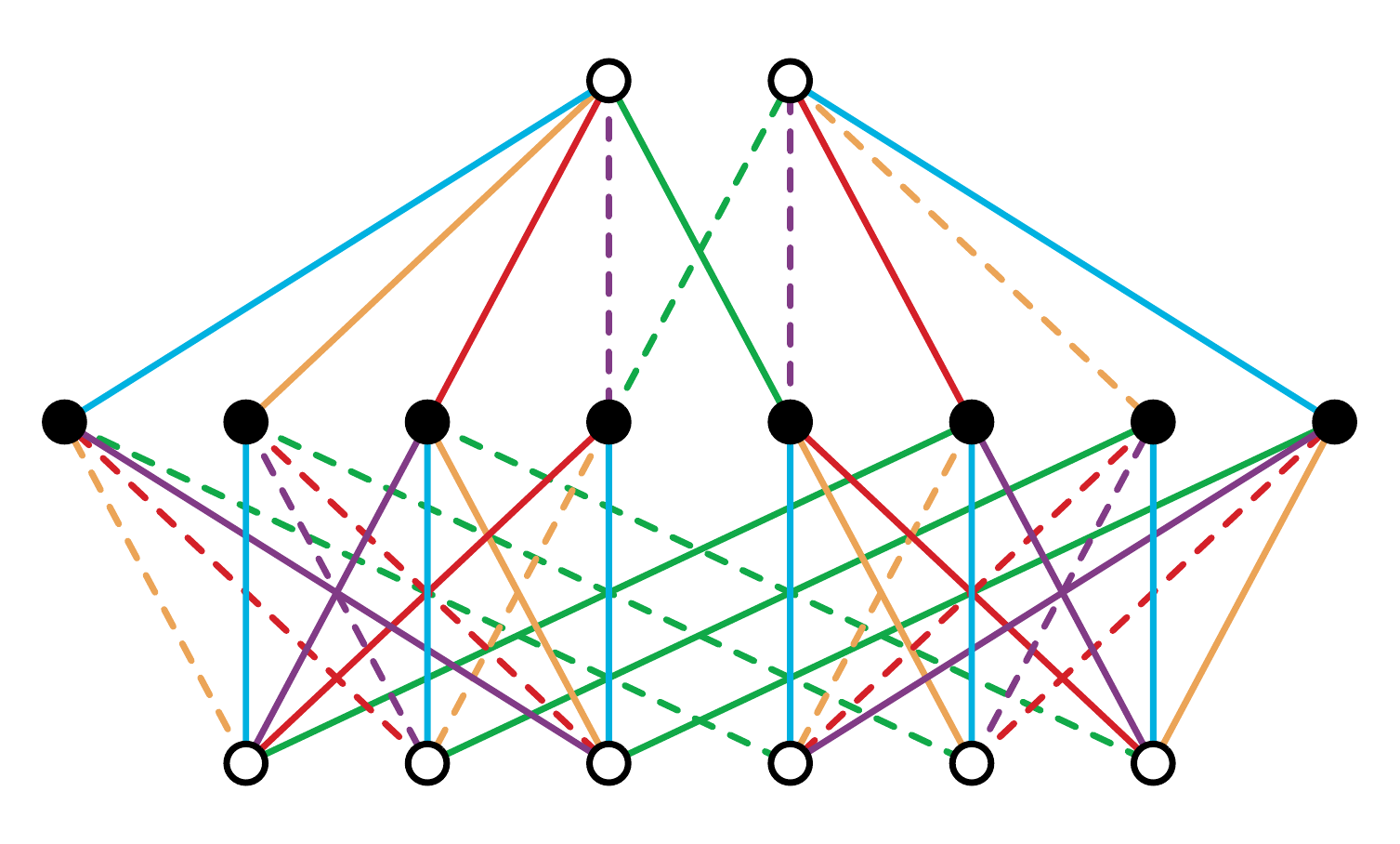} \hspace{1cm} \includegraphics[width=9.5cm]{ex1_2.pdf}
\end{center}

Applying the methods of Section \ref{sec:iso} to this pair proceeds as in the analysis of the Adinkras of Figure \ref{fig:cert1}. As in that example, after ordering the edge colors from blue to green, each Adinkra has associated code $(11110)$. Following the methods of Section \ref{sec:main}, we see that the top two nodes of the first Adinkra belong to different orbits of the corresponding 1-level Adinkra, whereas the top two nodes of the second Adinkra belong to the same orbit. Hence the two Adinkras are trivially distinguished. In particular, following the conventions in the equations of \ref{ex:orbits}, the $\mu$ values of these two Adinkras, labeled by $G$ and $H$ respectively, are:

\begin{align}
 &\mu_G(3,1) = 1 & & \mu_H(3,1) = 2 \\
 &\mu_G(3,2) = 1 & & \mu_H(3,2) = 0 \\
 &\mu_G(2,1) = \mu_G(2,2) = 4 & & \mu_H(2,1) = \mu_H(2,2) = 4 \\
 &\mu_G(1,1) = 3 & & \mu_H(1,1) = 2 \\
 &\mu_G(1,2) = 3 & & \mu_H(1,2) = 4
\end{align}

\noindent Hence $\textrm{cert}_G \ne \textrm{cert}_H$, and the two Adinkras are in different isomorphism classes. Also, since the `orbit spread', the number of vertices in each orbit at each height, has a different character in each Adinkra regardless of the ordering of the orbits, these two Adinkras will still remain in different isomorphism classes if edge-color permutations are allowed.

The second pair of Adinkras in \cite{counter} comprises two height 3, $(6,2)$ Adinkras, isomorphic to the pair displayed below.

\begin{center}
 \includegraphics[width=9.5cm]{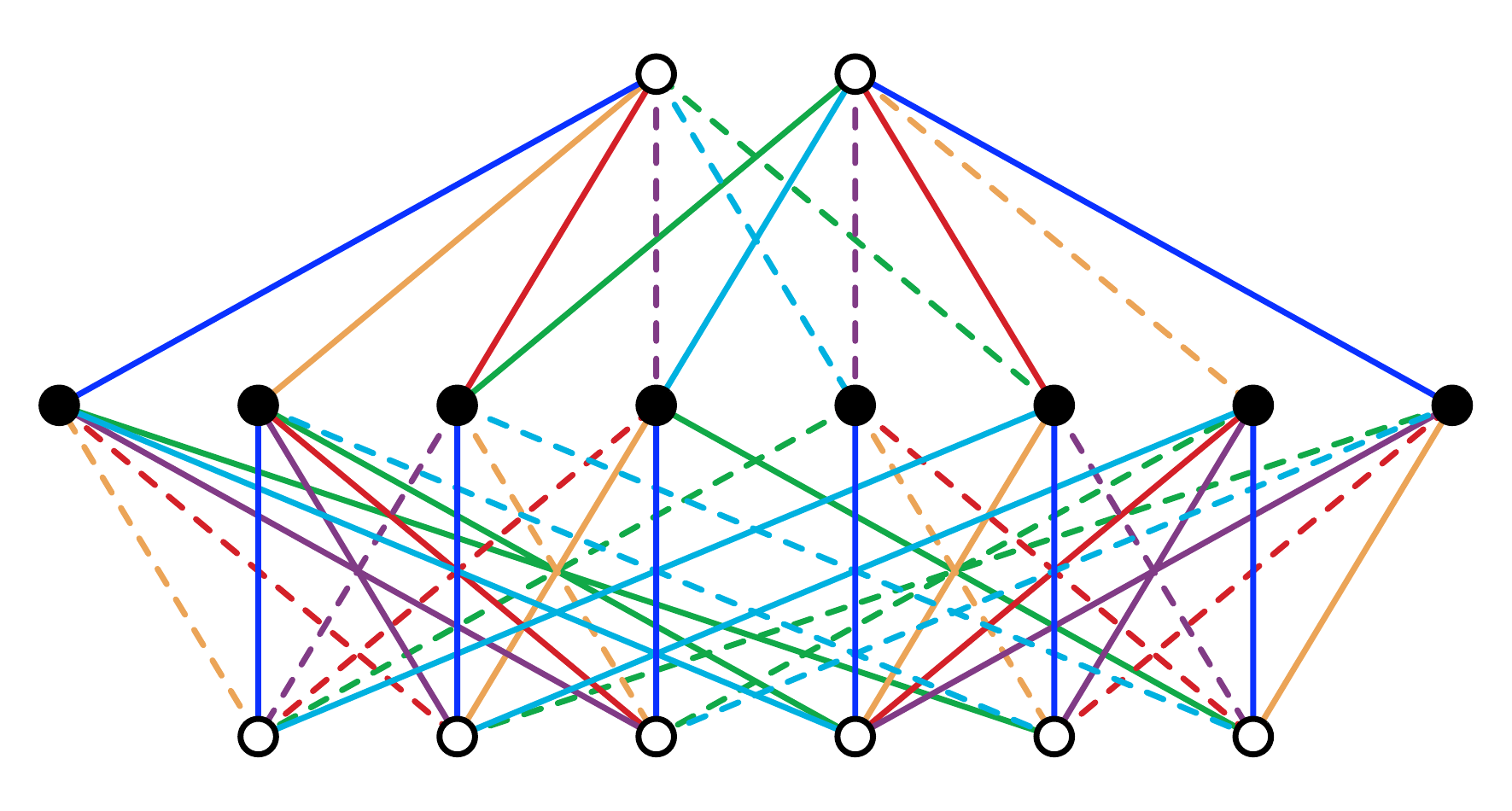} \hspace{1cm} \includegraphics[width=9.5cm]{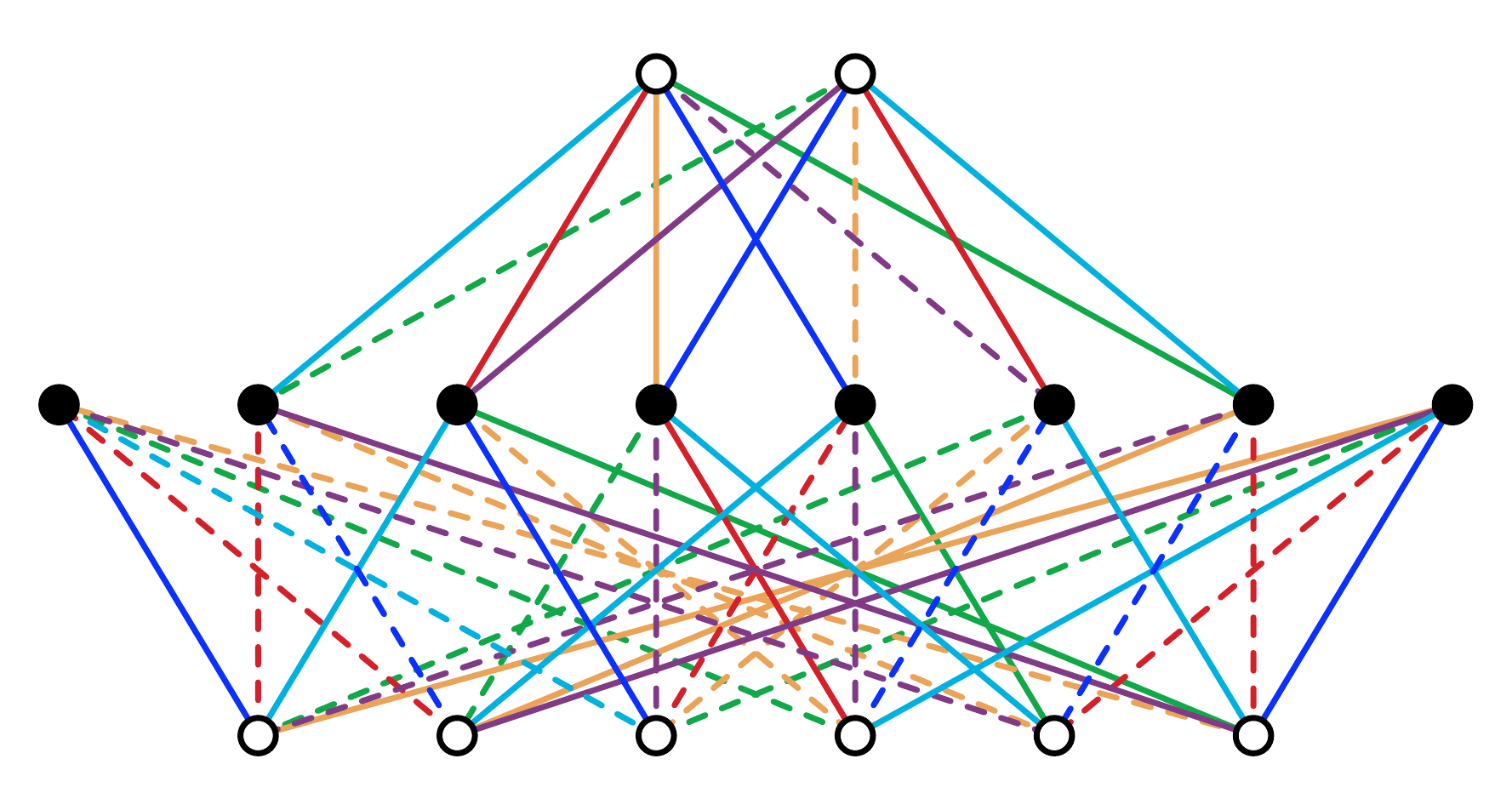}
\end{center}

In this case, ordering the edge colors from dark blue to green, as for the top left node of the first Adinkra), both Adinkras have the associated code generated by
$\begin{pmatrix}
 1\,1\,1\,1\,0\,0\\[-1mm]
 0\,0\,1\,1\,1\,1\\[-1mm]
\end{pmatrix}$, and hence each have four orbits in the automorphism group of the corresponding 1-level Adinkras.
The top two nodes of the first Adinkra are connected via the length-2 paths $(001001)$ and $(000110)$, having odd inner product with each of the codewords in the generating set above. Hence these nodes are in different orbits of the automorphism group. Conversely, the top two nodes of the second Adinkra are connected via the length-2 paths $(110000)$, $(001100)$ and $(000011)$, having even inner product with each of the codewords. Hence they in the same orbit. As a result, the certificates described in Section \ref{sec:iso} are different for each Adinkra, hence they are in different isomorphism classes. Again, following the conventions in the equations of \ref{ex:orbits}, we obtain $\mu$ values of:

\begin{align}
 &\mu_G(3,1) = \mu_G(3,2) = 1 & & \mu_H(3,2) = 2 \\
 &\mu_G(2,i) = 2, \quad \forall \; i \in \left[4\right] & & \mu_H(2,i) = 2, \quad \forall \; i \in \left[4\right] \\
 &\mu_G(1,1) = \mu_G(1,2) = 1 & & \mu_H(1,2) = 0 \\
 &\mu_G(1,3) = \mu_G(1,4) = 2 & & \mu_H(1,1) = \mu_G(1,3) = \mu_G(1,4) = 2
\end{align}

\noindent Hence, as in the previous example, $\textrm{cert}_G \ne \textrm{cert}_H$, and moreover the Adinkras remain in different isomorphism classes if edge-color permutations are allowed.

\newpage

\def\rasp{\leavevmode\raise.45ex\hbox{$\rhook$}}

\end{document}